%% file: kernel_mmd.tex
\documentclass{article}

\input{preamble}

\usepackage[preprint]{neurips_2023}

\usepackage[utf8]{inputenc} 
\usepackage[T1]{fontenc}    
\usepackage{url}            
\usepackage{booktabs}       
\usepackage{amsfonts}       
\usepackage{nicefrac}       
\usepackage{microtype}      
\usepackage{xcolor}         

\DeclareMathOperator{\mmd}{W}
\DeclareMathOperator{\TV}{TV}
\newcommand{\inv}[1]{\left({#1}\right)^{-1}}
\newcommand{\what}[1]{\widehat{#1}}
\newcommand{\wtil}[1]{\widetilde{#1}}
\newcommand{\subs}[2]{\left[{#1}\right]_{#2}}
\newcommand{\supnorm}[1]{\norm{#1}_{\infty}}
\newcommand{\opnorm}[1]{\norm*{#1}_{\mathrm{op}}}

\newenvironment{proof}{\paragraph{Proof:}}{\hfill$\square$}

\begin{document}
	
\title{Semi-Nonparametric Estimation of Distribution Divergence in Non-Euclidean Spaces}

%

\author{%
  Chong Xiao Wang, Wee Peng Tay, Zhenghua Chen \\
  School of Electrical and Electronic Engineering \\
  Nanyang Technological University\\
  50 Nanyang Ave, Singapore 639798 \\
  \texttt{wang1216@e.ntu.edu.sg,wptay@ntu.edu.sg,chen0832@e.ntu.edu.sg} \\
}

\maketitle

\begin{abstract}
	This paper explores methods for estimating or approximating the total variation distance and the chi-squared divergence of probability measures within topological sample spaces, using independent and identically distributed samples. Our focus is on the practical scenario where the sample space is homeomorphic to subsets of Euclidean space, with the specific homeomorphism remaining unknown. Our proposed methods rely on the integral probability metric with witness functions in universal reproducing kernel Hilbert spaces (RKHSs). The estimators we develop consist of learnable parametric functions mapping the sample space to Euclidean space, paired with universal kernels defined in Euclidean space. This approach effectively overcomes the challenge of constructing universal kernels directly on non-Euclidean spaces. Furthermore, the estimators we devise demonstrate asymptotic consistency, and we provide a detailed statistical analysis, shedding light on their practical implementation.
\end{abstract}


\section{Introduction}
Quantifying the divergence of probability measures from observed samples holds pivotal roles across various machine learning domains, often serving as a metric to optimize, such as in applications like the variational auto-encoder \cite{KinWel:C2014} and information bottleneck \cite{AleFisDil:C17}. This paper delves into methods for estimating two classical $f$-divergences --- namely the total variation (TV) distance and $\chi^2$-divergence. 

Despite the availability of a plethora of parametric and non-parametric estimators for these divergences, many off-the-shelve options impose restrictive assumptions on the domain knowledge of the underlying model or necessitate density approximation techniques like space partitioning. For instance, methods such as the $k$-nearest-neighbor ($k$-NN) distances method \citep{WanKulVer:J09,SinPoc:C16} and the nearest neighbor ratio (NNR) estimator \citep{NosMooSek:C17} assume multidimensional continuous densities in the Euclidean space. Additionally, the estimators derived from the von Mises expansion in \citet{KanKriPoc:C15} require density estimation. In contrast, parametric approaches such as the variational method based on the evidence lower-bound \citep{BleKucMcA:J17} and neural estimators based on the variational representation of $f$-divergences \citep{BelBarRaj:C18}, while not stipulating the topology of the input data, prove sensitive to the chosen modeling frameworks, potentially leading to consistency issues \citep{GhiMasDy:C21}.

\paragraph{Motivation}
To summarize, distribution divergence has been predominantly explored within Euclidean spaces, and tractable computation or estimation often leads to strong simplifications. The pursuit of an analytically computed, theoretically reliable, and practically implementable distribution divergence in non-Euclidean spaces remains a formidable challenge. Nevertheless, studying divergence in topological spaces holds great significance across various domains. In the contemporary landscape of machine learning, where complex data structures and manifold learning have become prominent, there arises a compelling need for divergences that are specifically tailored to smooth manifolds. In numerous machine learning applications with such configurations, high-dimensional observations inherently admits a low-dimensional latent representation in Euclidean spaces (cf. the references \citep{KinWel:C2014,AleFisDil:C17}). Therefore, incorporating domain-specific knowledge between the sample space and Euclidean space can help address the challenges associated with non-Euclidean spaces for developing more nuanced and effective approaches to quantify distribution divergence in diverse and intricate data scenarios.

\paragraph{Our approach}
Integral probability metrics (IPMs), as discussed in \citep{Mul:J97}, measures divergence by evaluating the largest difference in expectation between two probability distributions over a class of measurable functions known as witness functions. In this paper, we capitalize on IPMs to derive the TV distance and $\chi^2$-divergence on a topological sample space $\Omega$. This is achieved by identifying two sub-classes within the continuous function space as witness functions, as concluded in \cref{thm:KLDbounds}. 

To make the problem tractable, we introduce the concept of reproducing kernel Hilbert space (RKHS) for the witness functions in IPMs. The choice of kernel is of paramount importance, and we specifically aim for them to be \emph{universal} \citep{SriFukLan:J11,SimSch:J18}. Intuitively, the universality ensures that each distribution admits a unique embedding in RKHSs \citep{MuaFukSri:J17}. While universal kernels exist off-the-shelve for Euclidean space, such as polynomial kernels and Gaussian radial basis function (RBF) kernels \citep{Ste:J02}, specifying an exact form of universal kernels is generally intractable for topological spaces. Addressing this, we extend the universal kernels from Euclidean space to topological spaces, as proposed in \cref{props:ker_composition}. This proposition asserts that the composition of a universal kernel on $\bbR^d$ and a homeomorphism from $\Omega$ to a Euclidean subset is a universal kernel on $\Omega$. 

However, establishing a homeomorphism is a highly challenging task, analytically or empirically, especially when the topology of the input data is unknown. For example, in the context of smooth manifolds \citep{Lee:B12}, the Whitney embedding theorem \citep{Lee:B12} guarantees the existence of a homeomorphism between any $N$-dimensional smooth manifold and some subsets of $\bbR^M$ for some $M\leq 2N$. Nevertheless, there is no known result demonstrating how to identify a homeomorphism solely based on the smooth structures. In response to this difficulty, our paper proposes an effective way to learn a composition kernel on $\Omega$ for the estimation of the TV distance and $\chi^2$-divergence. This provides a practical and robust solution to extend the application of universal kernels to diverse spaces, including topological manifolds.

\paragraph{Related works} 
Kernel embedding, considered one of the most widely studied non-parametric measures of nonlinear dependencies, has garnered substantial attention in the literature \citep{SmoGreSon:C07,SriGreFuk:J10}. Notable related works encompass kernel maximum mean discrepancy (MMD) (also known as the kernel two-sample test) \citep{GreBor:J06}, kernel canonical correlation \citep{BacJor:J02,FukBacGre:J07}, Hilbert-Schmidt norm of covariance operators \citep{GreBouSmo:C05,GreFukTeo:C07}, and measure of conditional dependence based on normalized cross-covariance operators on RKHSs \citep{FukGreSun:C07,GreHerSmo:J05}.  The success of these approaches stems from two key factors. Firstly, embedding using a characteristic kernel uniquely preserves all information about a distribution \citep{FukGreLan:C09}. Secondly, computations on the potentially infinite-dimensional RKHSs can be efficiently carried out through simple Gram matrix operations, thanks to the kernel trick \citep{VerTsuSch:J04}. 

Applying the aforementioned approaches to manifold settings requires kernels specifically tailored for manifolds, an area of exploration that is still in its early stages. The reference \citep{CheXie:J21} characterized the test power of the kernel MMD on a Riemannian submanifold of Euclidean spaces, considering density functions (with respect to the Riemannian volume) that belong to certain classes related to geodesics. Another work \citep{OzaGra:C09} proposed estimating kernel density on a Riemannian manifold using kernels on the Euclidean space with small bandwidth, arguing that the geodesics measured on the manifold closely resemble Euclidean distances when points are in close proximity. The work \citep{JayHarSal:J15} identified positive-definite geodesic exponential kernels on two specific manifolds. However, \citet{FerHau:C16} pointed out that geodesic exponential kernels are only positive-definite for all bandwidths when the input space has strong linear properties. None of these works generalize to kernel methods in topological spaces and ensure universality.

\paragraph{Contributions}
Assuming a topological sample space $\Omega$ that is homeomorphic to certain Euclidean subsets, this paper introduces a practical kernel method for estimating the TV distance and $\chi^2$-divergence for probability measures over $\Omega$. To accomplish this, we propose the learning of a composition kernel, equivalent to universal kernels, from a parametric family of continuous functions mapping from $\Omega$ to $\bbR^d$ as illustrated in \cref{fig:layout}. What distinguishes our work from existing studies is the simultaneous estimation of both the TV distance and the $\chi^2$-divergence, employing a shared learnable kernel, subject to specific technical conditions. Moreover, we conduct a comprehensive statistical analysis of the estimators, providing insights to the choice of hyper-parameters.
\begin{figure}[!htb]
	\centering
	\includegraphics[scale=0.4]{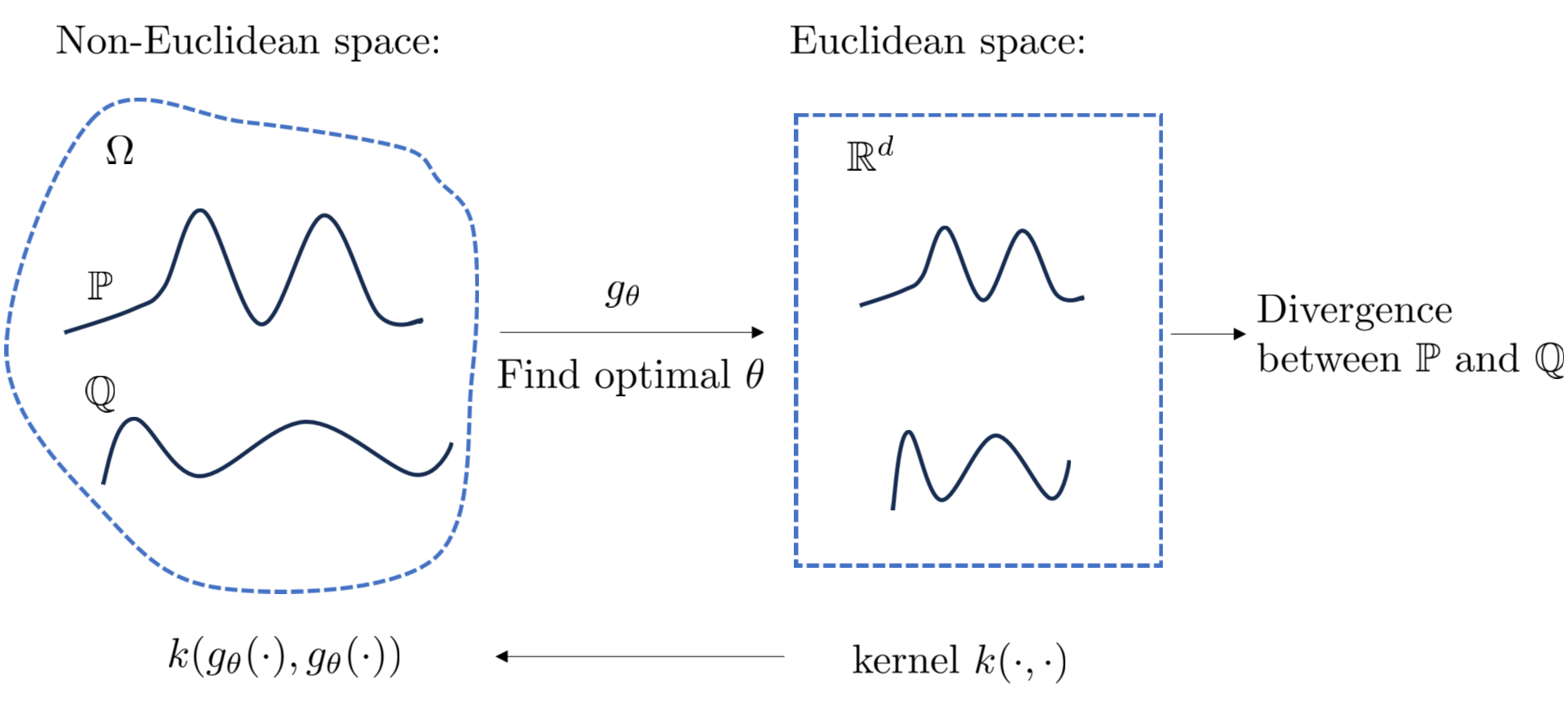}
	\caption{Estimating the divergence between distributions $\bbP$ and $\bbQ$ within a non-Euclidean sample space $\Omega$ through learning a function from $\Omega$ to $\bbR^d$.}
	\label{fig:layout}
\end{figure}

The rest the paper is organized as follows. In \cref{sect:rkhs}, we provide an introduction to RKHSs and universal kernels. Moving on to \cref{sect:ipm_div}, we express the TV distance and $\chi^2$-divergence as IPMs using witness functions within RKHSs. Following this, \cref{sect:ker_div} details the derivation of analytical RKHS formula for the TV distance and $\chi^2$-divergence. This section also introduces a method to learn a transformation function from $\Omega$ to $\bbR^d$ to form a desirable kernel on $\Omega$. In \cref{sect:ker_div_est}, we shift our focus to sample-based kernel estimates for the proposed IPMs, accompanied by a discussion on error probability and asymptotic analysis in the large sample size regime. Finally in \cref{sect:expt}, we test the derived estimators by applying them to an image dataset.

\paragraph{Notations} 
We use $\circ$ to denote function composition and $\otimes$ to denote the tensor product between two functions. Given a probability measure $\bbP$ and random variable $f$, we use $\E_{\bbP}[f] \coloneqq\int f \ud\bbP$ to denote the expectation of $f$ under $\bbP$. The sets of real and natural numbers are $\bbR$ and $\bbN$, respectively. We denote $\ip{}{}$ and $\norm{\cdot}$ as the dot product and Euclidean norm in the Euclidean space. Absolute value is denoted by $\abs{\cdot}$ and $\bA\T$ denotes the transpose of matrix $\bA$. We let $\exp(\cdot)$ be the exponential function. The logarithm $\log$ is natural logarithm. We denote the RKHS associated with a kernel by the calligraphic capital letter of the kernel symbol, e.g., RKHS $\calK$ associated with kernel $k(\cdot,\cdot)$. We denote $\ip{}{}_{\calK}$ and $\norm{\cdot}_{\calK}$ as the kernel inner product and kernel norm in the RKHS $\calK$. We denote $\opnorm{\cdot}$ as the operator norm in RKHS $\calK$, i.e., $\opnorm{\Psi}=\inf\{c\geq 0, \norm*{\Psi f}_{\calK}\leq c\norm*{f}_{\calK}\}$. We let $\bI_n$ be the $n\times n$ identity matrix and $I$ be the identity map. 

\section{Preliminaries}\label{sect:rkhs}
In this section, we revisit the notions of a RKHS and universal kernels. For an in-depth literature review on RKHS, we recommend referring to the works \citep{ManAmb:J15}. Additionally, comprehensive insights into universal kernels can be found in the works \citep{Ste:J02,MicXuZha:J06,SriFukLan:J11}.

\subsection{Review of RKHS}\label{appx:review_RKHS}
Let $(\calE,\scT)$ be a topological space and consider a Hilbert space $\calK$ of functions from $\calE$ to $\bbR$, with the inner product denoted as $\ip*{\cdot}{\cdot}_{\calK}$. The Hilbert space $\calK$ is a RKHS if at each $x\in\calE$, the point evaluation operator $L_x(f) = f(x)$ for $f\in\calK$ is a bounded linear functional, i.e., there exists a positive constant $C$ such that $\norm*{L_x(f)}_{\calK}\leq C\norm*{f}_{\calK}$ for all $f\in\calK$. Then, from the Riesz representation theorem \citep{DebMik:B08}, for each $x\in\calE$, there exists a unique element $\Phi(x)\in\calK$ such that
\begin{align*}
	f(x)=\ip*{f}{\Phi(x)}_{\calK},
\end{align*}
for any $f\in\calK$. This is called the reproducing property. For any $y\in\calE$, we have $\Phi(x)(y)=\ip*{\Phi(x)}{\Phi(y)}_{\calK}$. The kernel associated with $\calK$ is a function $\calE\times\calE\to\bbR$ defined as
\begin{align*}
	k(x,y)
	=\ip*{\Phi(x)}{\Phi(y)}_{\calK},
\end{align*}
for all $x,y\in\calE$. Therefore, $\Phi(x)$ as a function defined on $\calE$ can be written as $\Phi(x)(\cdot)=k(x,\cdot)$. It is easy to see that a kernel is symmetric and positive definite, i.e., $\sum_{i,j=1}^nc_ic_jk(x_i,x_j)\geq 0$ for any $n\in\bbN$, $x_i\in\calE$ and $c_i\in\bbR$ for $ i=1,\ldots,n$. The RKHS can be written as the completion of the linear span of the feature maps in $\calK$:
\begin{align*}
	\calK
	=\overline{\spn}\left\{\Phi(x): x\in\calE\right\}.
\end{align*}

On the other hand, from the Moore-Aronszajn theorem \citep{Aro:J50}, such a positive definite kernel $k(\cdot,\cdot)$ uniquely determines a Hilbert space from $\calE$ to $\bbR$, for which $k(\cdot,\cdot)$ is a reproducing kernel.
In general, a feature map is a map $\phi:\calE\to\calF$, where $(\calF,\ip*{\cdot}{\cdot}_{\calF})$ is a Hilbert space which we will call the feature space. Every feature map defines a kernel via
\begin{align}\label{ker_feat_ip}
	k(x,y)=\ip*{\phi(x)}{\phi(y)}_{\calF}.
\end{align}
Letting $\calF=\calK$, it is obvious that $\Phi(x)$ is a feature map of $x$ in $\calK$. We call it the \emph{canonical feature map}. A RKHS has infinitely many associated feature maps.

\subsection{Universal Kernels}
For a kernel $k:\calE\times\calE\to\bbR$, we denote $\Phi(x)$ as the canonical feature map for $x\in\calE$, defined as $\Phi(x)(\cdot)=k(x,\cdot)$. We begin with a brief introduction to the universal kernels \citep{Ste:J02,MicXuZha:J06,SriFukLan:J11}. Suppose $\calK$ is a RKHS associated with reproducing kernel $k(\cdot,\cdot)$, defined on $\calE\times\calE$. Let $\Omega$ be a compact subset of $\calE$ and $\calK(\Omega)$ be a closed subspace of $\calK$, for which the domain of the functions in $\calK$ is restricted to $\Omega$:
\begin{align}\label{eq:ker_space}
	\calK(\Omega)
	=\overline{\spn}\left\{\evalat{k(x, \cdot)}_\Omega : x\in\calE\right\}.
\end{align} 
Suppose that for any compact subset $\Omega\subset\calE$, the space $\calK(\Omega)$ is dense in $C(\Omega)$, i.e., for any $\epsilon> 0$ and any function $f\in C(\Omega)$, there exists a function $g\in\calK(\Omega)$ such that $\supnorm{f-g}\leq\epsilon$. This is called the universal approximation property. In other words, universal RKHSs can uniformly approximate any continuous function defined on a compact space. A kernel with this property is called a \emph{universal kernel} and the RKHS it induces is said to be universal. 

The reference \citep{Ste:J02} has shown in Corollary 10 and 11 that kernels that can be expanded into its Taylor series or in a pointwise convergent Fourier series with positive coefficients are universal w.r.t. the compact subsets of the Euclidean space $\bbR^d$. Concrete examples of such kernels include the Gaussian radial basis function (RBF) kernels like
\begin{subequations}\label{eq:unv_kernels}
	\begin{align} 
		&k(x,y)=\exp\parens*{-\gamma\norm*{x-y}^2},\label{unv_ker_a} \\ 
		\intertext{and infinite polynomial kernels like}
		&k(x,y)=\left(1+\norm*{x-y}^2\right)^{-\gamma},\label{unv_ker_b}
	\end{align}
\end{subequations}
where $x,y\in\bbR^d$, $\gamma>0$, and $\norm*{}$ is the Euclidean norm. The kernels \cref{unv_ker_a,unv_ker_b} are universal on every compact subset of $\bbR^d$. For an insight into the universal approximation property, consider the case $N=1$ and the kernel \cref{unv_ker_a}, which has a feature map in $\ell^2$ space (the space of square-summable sequences) as
\begin{align*}
	\phi(x)
	=\exp(-\gamma x^2)\left[1,\sqrt{\frac{2\gamma}{1!}}x,\sqrt{\frac{(2\gamma)^2}{2!}}x^2,\ldots\right]\T,
\end{align*}
for $x\in\bbR$. The set of all feature maps is the spanning set of the $\ell^2$ space. The canonical feature space $\calK(\Omega)$ can be represented as $\left\{f(x)=\ip*{\phi(x)}{g} : x\in\Omega,g\in\ell^2\right\}$, in which $\ip{}{}$ is the dot product. Therefore, $\calK(\Omega)$ is composed of all infinite degree polynomial functions, which can approximate any continuous functions defined over a compact set, according to the celebrated Stone-Weierstrass approximation theorem \citep{Dud:B02}.

\subsection{Composition Kernel}
Here, we wish to highlight our notation convention: given a kernel $k(\cdot,\cdot)$, we denote the RKHS it induces by the calligraphic capital letter of the kernel symbol, i.e., $\calK$. The kernel inner product is denoted as $\ip*{\cdot}{\cdot}_{\calK}$.

To circumvent the challenging task of explicitly constructing a universal kernel on arbitrary topological spaces, we turn to composition kernels. As depicted in \cref{props:ker_composition}, universal kernels on any topological spaces can be formed by composing a mapping function and universal kernels on other topological spaces.
\begin{Proposition}\label{props:ker_composition}
	Let $(\Omega,\scT)$ and $(\Omega',\scT')$ be two topological spaces. Let $C(\Omega')$ be the space of bounded continuous functions from $\Omega'$ to $\bbR$. Suppose $\Omega$ is compact and $\calK(\Omega')$ is dense in $C(\Omega')$ under the uniform norm. If $g:\Omega\to\Omega'$ is a continuous and bijective mapping, then the set $\left\{f\circ g : f\in\calK(\Omega')\right\}$ is dense in $C(\Omega)$.
\end{Proposition}
The proof is provided in \cref{appx:ker_comp}.

The implication of \cref{props:ker_composition} is that for any topological space $\Omega$ homeomorphic to a subset of Euclidean space, the composition of a homeomorphism $g:\Omega\to g(\Omega)\subseteq\bbR^d$ and a universal kernel $k(\cdot,\cdot)$ in $\bbR^d$ yields a kernel $k(g(\cdot),g(\cdot))$ that is universal on $\Omega$. The remaining challenge is to identify a suitable mapping function $g$ for constructing a kernel on $\Omega$.

\section{IPMs and Divergence}\label{sect:ipm_div}
In this section, we introduce two specific classes of witness functions for IPMs in continuous function space, leading to the TV distance and $\chi^2$-divergence, respectively.

We begin with reviewing the definition of IPM and $f$-divergence. Consider a \emph{compact} topological space $(\Omega,\scT)$, where $\scT$ is a topology on $\Omega$. The Borel $\sigma$-algebra $\calB(\Omega)$ associated with $\Omega$ is a measurable space generated by the open subsets of $\Omega$. Note any space with a topology has a Borel $\sigma$-algebra. Let $\bbP$ and $\bbQ$ be two Borel probability measures over the measurable space $(\Omega, \calB(\Omega))$ satisfying the following assumption. 

\begin{Assumption}\label{asm:abs_conts}
	The probability measure $\bbP$ is absolutely continuous w.r.t. $\bbQ$, i.e., for any $A\in\calB(\Omega)$ such that $\bbQ(A)=0$, we have $\bbP(A)=0$. 
\end{Assumption}
With \cref{asm:abs_conts}, the Radon-Nikodym theorem \citep[Theorem 5.5.4]{Dud:B02} says that there exists a non-negative real-valued integrable function $f:\Omega\to\bbR$ such that
\begin{align*}
	\P(A)=\int_{A}f \ud\bbQ,\ \forall A\in\calB(\Omega).
\end{align*}
The function $f$ is called the Radon-Nikodym derivative and is denoted as $\ddfrac{\bbP}{\bbQ}$.

The $f$-divergence originated from information theory \citep{CovTho:B05} measures how one probability measure is different from a reference probability measure and is defined as
\begin{align*}
	\KLD{\bbP}{\bbQ}=\int_{\Omega}f\left(\ddfrac{\bbP}{\bbQ}\right) \ud\bbP
\end{align*}
under \cref{asm:abs_conts}, where $f:[0,\infty)\to\bbR$ is a a convex function such that $f(1)=0$. The $f$-divergence is non-negative and zero if and only if $\bbP=\bbQ$ almost surely. The TV distance $\TV(\bbP,\bbQ)$ and chi-squared divergence $\chi^2(\bbP,\bbQ)$ correspond to $f(x)=\ofrac{2}\abs{x-1}$ and $f=(x-1)^2$, respectively. Note that the $f$-divergence may not be a metric on the space of probability measures since it can be asymmetric, i.e., $\chi^2(\bbP,\bbQ)\neq\chi^2(\bbQ,\bbP)$. 

On the other hand, the IPM, originally investigated by \citet{Mul:J97} is an integral probability metric that finds the maximal difference between the mean value of a set of witness functions w.r.t. two probability measures.
\begin{Definition}\label{defn:mmd}
	Let $\calF$ be a class of functions $f:\Omega\to\bbR$ measurable w.r.t. $\calB(\Omega)$. The integral probability metric (IPM) between two probability measures $\bbP$ and $\bbQ$ is
	\begin{align}\label{eq:mmd}
		\mmd[\calF,\bbP,\bbQ]=\sup_{f\in\calF}\left(\E_{\bbP}[f]-\E_{\bbQ}[f]\right).
	\end{align}
\end{Definition}

The functions in $\calF$ are called witness functions since they witness the difference in distributions. Obviously, the quality of the IPM as a statistic depends on the class $\calF$ of functions that define it. It has been shown in \citep[Lemma 9.3.2]{Dud:B02} that $\bbP=\bbQ$ if and only if $\E_{\bbP}[f]=\E_{\bbQ}[f]$ for all bounded continuous functions $f:\Omega\to\bbR$. Therefore, if $\calF$ includes the space of bounded continuous functions, the IPM is a metric on the space of probability measures. 

However, the space of bounded continuous functions is too large to be manageable. It is of great interest to identify proper subsets of witness functions that can lead to desirable properties. 

Let $C(\Omega)$ be the space of continuous functions from $\Omega$ to $\bbR$. Since $\Omega$ is compact, all continuous functions are bounded, and we denote
\begin{align*}
	C_{\infty}(\Omega)=\left\{f\in C(\Omega) : \norm*{f}_{\infty}\leq 1\right\}, 
\end{align*}
where $\supnorm{f}=\max_{x\in\Omega}\abs{f(x)}$ is the uniform norm. The set $C_\infty(\Omega)$ is a possible choice of $\calF$ in \cref{defn:mmd}.

If the Borel space $\calB(\Omega)$ is assigned a \emph{finite} regular measure $\nu$ (i.e., $\nu(\Omega)<\infty$), we may define an inner product for the elements in $C(\Omega)$ as
\begin{align*}
	\ip*{f}{g}_{\nu}=\int_{\Omega}fg \ud\nu,\ \forall f,g\in C(\Omega).
\end{align*}
This produces a Hilbert space $\left(C(\Omega),\ip*{\cdot}{\cdot}_{\nu}\right)$ of functions from $\Omega$ to $\bbR$. A norm operator $\norm*{\cdot}_{\nu}=\sqrt{\ip*{\cdot}{\cdot}_{\nu}}$ is naturally induced by the inner product associated with the Hilbert space. Then we may let $\calF$ in \cref{defn:mmd} be the unit ball of the square-integrable functions in $C(\Omega)$ w.r.t. measure $\nu$, denoted as
\begin{align}\label{eq:cls_l2_norm}
	C_2(\Omega,\nu)=\left\{f\in C(\Omega) : \norm*{f}_{\nu}\leq 1\right\}.
\end{align}
Please distinguish the inner product $\ip*{\cdot}{\cdot}_{\nu}$ (where $\nu$ is a measure) in \cref{eq:cls_l2_norm} from the kernel inner product $\ip*{\cdot}{\cdot}_{\calK}$ in the following contexts to avoid confusions. Note that $C_{\infty}(\Omega)$ is a subset of $C_2(\Omega,\nu)$ when $\nu$ is a normalized measure with $\nu(\Omega)=1$. 

By choosing $\calF=C_{\infty}(\Omega)$ or $\calF=C_2(\Omega,\bbQ)$ for $\mmd[\calF,\bbP,\bbQ]$, we obtain $\TV(\bbP,\bbQ)$ and $\chi^2(\bbP,\bbQ)$, as specified in \cref{thm:KLDbounds}.
\begin{Theorem}\label{thm:KLDbounds}
	Suppose \cref{asm:abs_conts} holds and assume (a $\bbQ$-version of) $\ddfrac{\bbP}{\bbQ}$ is continuous on $\Omega$. We have
	\begin{align} \label{eq:mmd_tv}
		&\mmd[C_{\infty}(\Omega),\bbP,\bbQ]=\TV(\bbP,\bbQ),\\ \label{eq:mmd_chi2}
		&\mmd^2[C_2(\Omega,\bbQ),\bbP,\bbQ]=\chi^2(\bbP,\bbQ).
	\end{align}
\end{Theorem}
The proof is provided in \cref{appx:ipm_tv_chi2}. 

The IPMs $\mmd[C_{\infty}(\Omega),\bbP,\bbQ]$ and $\mmd[C_2(\Omega,\bbQ),\bbP,\bbQ]$ are the objects of our study.

\subsection{Witness Functions in RKHS}\label{sect:mmd_in_rkhs}
However, working with the witness function spaces $C_2(\Omega,\bbQ)$ and $C_{\infty}(\Omega)$ for computing IPMs \cref{eq:mmd_tv} and \cref{eq:mmd_chi2} is computationally infeasible. An alternative approach involves employing the RKHSs associated with universal kernels for the witness functions of the IPMs. By replacing the continuous function space from $\Omega$ to $\bbR$ with an RKHS $\calK$ (restricted to $\Omega$), we derive surrogate witness functions in the RKHS as replacement for the witness functions $C_{\infty}(\Omega)$ and $C_2(\Omega,\bbQ)$, respectively:
\begin{align}\label{eq:cls_max_norm_ker}
	&\calK_{\infty}(\Omega)=\left\{f\in\calK(\Omega) : \supnorm{f}\leq 1\right\} \\ \label{eq:cls_l2_norm_ker}
	&\calK_2(\Omega,\nu)=\left\{f\in\calK(\Omega) : \norm*{f}_{\nu}\leq 1\right\},
\end{align}

As demonstrated in \cref{cor:CK}, it is without loss of generality to work with witness functions in universal RKHSs to compute the IPMs \cref{eq:mmd_tv,eq:mmd_chi2}.
\begin{Lemma}\label{cor:CK}
	If $\calK$ is a universal RKHS, we have
	\begin{align*}
		&\mmd[\calK_{\infty}(\Omega),\bbP,\bbQ]
		=\mmd[C_{\infty}(\Omega),\bbP,\bbQ]
		=\TV(\bbP,\bbQ),\nn
		&\mmd[\calK_2(\Omega,\nu),\bbP,\bbQ]
		=\mmd[C_2(\Omega,\nu),\bbP,\bbQ]
		=\chi^2(\bbP,\bbQ).
	\end{align*}
\end{Lemma}
In the subsequent sections, we present methods to compute and estimate $\mmd[\calK_{\infty}(\Omega),\bbP,\bbQ]$ and $\mmd[\calK_2(\Omega,\bbQ),\bbP,\bbQ]$ by leveraging the RKHS properties.

To begin with, we provide a summary of the notations frequently used in the subsequent context in \cref{tab:notations}. This aims to enhance readability and prevent potential confusions.
\begin{table}[!htbp]
	\centering
	\caption{List of Notations}\label{tab:notations}
	\begin{tabular}{ll} 
		\toprule
		\multicolumn{1}{c}{Symbols} & \multicolumn{1}{c}{Description} \\
		\midrule
		$k_{\theta}$ & Composition kernel $k(g_{\theta}(\cdot),g_{\theta}(\cdot))$, with $g_{\theta}$ parameterized by $\theta\in\Theta$.\\
		$M_k$, $M_{k_{\theta}}$ & Maximum value of kernel $k$ and $k_{\theta}$, respectively. \\
		$\mmd[\calF,\bbP,\bbQ]$ & IPM where $\calF$ is the witness function class. \\
		$\mmd^2[\calF,\bbP,\bbQ]$ & Square of $\mmd[\calF,\bbP,\bbQ]$. \\
		$\what{\mmd}[\calF,\bbP,\bbQ]$ & Empirical Estimation of $\mmd[\calF,\bbP,\bbQ]$. \\
		$\calK$, $\calK_{\theta}$ & RKHSs associated with kernel $k$ and $k_{\theta}$, respectively. \\
		$\calK_{\infty}(\Omega)$ & Unit ball with $L_1$-norm in $\calK$ (restricted to domain $\Omega$). \\
		$\wtil{\calK}(\Omega)$ & Ball of radius $1/\sqrt{M_k}$ with kernel norm in $\calK$, a proxy for $\calK_{\infty}(\Omega)$. \\
		$\calK_2(\Omega,\bbQ)$ & Unit ball with $L_2$-norm in $\calK$ (restricted to domain $\Omega$). \\ 
		$\calK_2^{\lambda}(\Omega,\bbQ)$ & $\calK_2(\Omega,\bbQ)$ with $\lambda$-regularization in kernel norm, a proxy for $\calK_2(\Omega,\bbQ)$. \\
		\bottomrule
	\end{tabular}
\end{table}

\section{Divergence under RKHS}\label{sect:ker_div}
In this section, we focus on the RKHS representations of $\mmd[\calK_{\infty}(\Omega),\bbP,\bbQ]$ and $\mmd[\calK_{\infty}(\Omega),\bbP,\bbQ]$. Although closed-form solutions do not exist, we establish a lower-bound for $\mmd[\calK_{\infty}(\Omega),\bbP,\bbQ]$, and derive an approximation for $\mmd[\calK_2(\Omega,\bbQ),\bbP,\bbQ]$ with arbitrarily small approximation error. The lower-bound also serves as an objective function for efficiently learning a transformation function from $\Omega$ to $\bbR^d$ to create a composition kernel.

In the rest of this paper, we make the following assumptions.
\begin{Assumption}\label{asm:homeomorphism}
	The topological space $\Omega$ is compact and homeomorphic to a subset of $\bbR^d$.
\end{Assumption}
\begin{Assumption}\label{asm:int_ker}
	The kernel $k(\cdot,\cdot)$ is continuous and bounded on $\Omega$.
\end{Assumption}

We make use of the notion of kernel embedding of density functions \citep{SmoGreSon:C07,SriGreFuk:J10}. Recall our convention that a RKHS $\calK$ is associated with a kernel $k(\cdot,\cdot)$. Suppose we have a RKHS $\calK$ associated with a universal kernel satisfying \cref{asm:int_ker}. 
For each probability measure $\bbP$, consider the following function from $\calK$ to $\bbR$:
\begin{align}\label{eq:mean_rkhs}
	f\in\calK \mapsto \E_{\bbP}[f(X)] =\E_{\bbP}[\ip*{f}{\Phi(X)}_{\calK}],
\end{align}
which is a bounded linear functional since $\E_{\bbP}[k(X,X)]<\infty$ from \cref{asm:int_ker}.
From the Riesz representation theorem \citep[Theorem 3.7.7]{DebMik:B08},  there exists a unique element $\mu_{\bbP}\in\calK$ such that 
$\E_{\bbP}[f(X)]=\ip*{f}{\mu_{\bbP}}_{\calK}$, where
\begin{align}\label{eq:mean_embed}
	\mu_{\bbP}\coloneqq\E_{\bbP}[\Phi(X)],
\end{align}
in which the expectation here is the Bochner integral \citep{HsiEub:15}.
In other words, the mean embedding of the probability measure $\bbP$ is the expectation under $\bbP$ of the canonical feature map. From \citep[Lemma 9.3.2]{Dud:B02}, a universal RKHS ensures that the mean embedding \cref{eq:mean_embed} from the space of probability measures to $\calK$ is an injective mapping. Thus, the probability measure is uniquely identified by an element in a universal RKHS $\calK$. 
For $f\in\calK$, the mean difference in terms of $\bbP$ and $\bbQ$ can be written as
\begin{align}\label{eq:mean_diff}
	\E_{\bbP}[f(X)]-\E_{\bbQ}[f(X)]
	=\ip*{f}{\mu_{\bbP}-\mu_{\bbQ}}_{\calK}. 
\end{align}

\subsection{Lower Bound of $\mmd[\calK_{\infty}(\Omega),\bbP,\bbQ]$}
\label{sect:ker_mmd_lb}
We derive the RKHS representation of $\mmd[\calK_{\infty}(\Omega),\bbP,\bbQ]$. It is noticed that the uniform norm in $\calK_{\infty}(\Omega)$, as defined in \cref{eq:cls_max_norm_ker}, is intractable when applying the kernel trick to derive a closed-form formula. To address this difficulty, we opt for working with a tractable subset of $\calK_{\infty}(\Omega)$, which yields a lower bound of $\mmd[\calK_{\infty}(\Omega),\bbP,\bbQ]$. This involves relaxing the uniform norm in $\calK$ as follows:
\begin{align}\label{eq:max_norm_ub}
	\begin{split}
		\supnorm{f}
		=\max_{x\in\Omega}\abs{f(x)}
		&=\max_{x\in\Omega}\abs{\ip*{f}{\Phi(x)}_{\calK}} \\
		&\leq\max_{x\in\Omega} \norm*{f}_{\calK} \sqrt{\ip*{\Phi(x)}{\Phi(x)}_{\calK}} \\
		&=\norm*{f}_{\calK}\max_{x\in\Omega}\sqrt{k(x,x)},
	\end{split}
\end{align}
where the inequality is due to the Cauchy-Schwarz inequality. Note that $\max_{x\in\Omega}k(x,x)$ is independent of the function $f\in\calK$ and let's denote
\begin{align}\label{eq:M}
	M_k=\max_{x\in\Omega}k(x,x).
\end{align}

In particular, for translation-invariant kernels such as the RBF kernels \cref{unv_ker_a}, $M_k$ is independent of $\Omega$.

Making use of the mean difference \cref{eq:mean_diff} and relaxation \cref{eq:max_norm_ub}, we obtain a lower bound of $\mmd[\calK_{\infty}(\Omega),\bbP,\bbQ]$:
\begin{align}
	\mmd[\calK_{\infty}(\Omega),\bbP,\bbQ]
	&=\sup_{f\in\calK} \braces*{\frac{\E_{\bbP}[f(X)]-\E_{\bbQ}[f(X)]}{\supnorm{f}}} \nn
	&\geq\ofrac{\sqrt{M_k}}\sup_{f\in\calK}\frac{\ip*{f}{\mu_{\bbP}-\mu_{\bbQ}}_{\calK}}{\norm*{f}_{\calK}}\nn
	&=\frac{1}{\sqrt{M_k}}\norm*{\mu_{\bbP}-\mu_{\bbQ}}_{\calK}. \label{mmd:uniform_norm}
\end{align}

The lower-bound \cref{mmd:uniform_norm} can be written as an IPM:
\begin{align}\label{eq:var_lb_mmd}
	\mmd[\wtil{\calK}(\Omega),\bbP,\bbQ]
	=\ofrac{\sqrt{M_k}}\norm*{\mu_{\bbP}-\mu_{\bbQ}}_{\calK},
\end{align}
where the witness function $\wtil{\calK}(\Omega)$ is a ball of radius $1/\sqrt{M_k}$ in RKHS $\calK$, defined as
\begin{align*}
	\wtil{\calK}(\Omega)
	=\left\{f\in\calK(\Omega) : \norm*{f}_{\calK}\leq\frac{1}{\sqrt{M_k}}\right\}.
\end{align*}

The tightness of the lower-bound in \cref{mmd:uniform_norm} is subject to the choice of the kernel $k$, and equality is attained if there exists a point $x\in\Omega$ such that $\Phi(x)\propto\mu_{\bbP}-\mu_{\bbQ}$ (the proof is provided in \cref{appx:supnorm_equality}). Our strategy is to maximize $\mmd[\wtil{\calK}(\Omega),\bbP,\bbQ]$ over a prameterized composition kernel to approximate $\mmd[\calK_{\infty}(\Omega),\bbP,\bbQ]$.


\subsection{Learnable Composition Kernel}\label{sect:composition_ker}
\cref{cor:CK} asserts that a universal RKHS $\calK$ ensures the equivalence $\mmd[\calK_{\infty}(\Omega),\bbP,\bbQ]=\TV(\bbP,\bbQ)$. However,  finding an explicit universal kernel for the sample space  $\Omega$ is challenging because $\Omega$ can be non-Euclidean. To address this difficulty, we firstly form a set of composition kernels on $\Omega$ as candidates:
\begin{align}\label{eq:composition_ker}
	\frakJ=\left\{k_{\theta}\coloneqq k(g_{\theta}(\cdot),g_{\theta}(\cdot))\mid \theta\in\Theta\right\},
\end{align}
where $g_{\theta}:\Omega\mapsto\bbR^d$ is a family of continuous functions from $\Omega$ to $\bbR^d$ parameterized by $\theta$, and $\Theta$ is its parameter space, and $k$ is a fixed universal kernel on Euclidean space $\bbR^d$.

By \cref{props:ker_composition}, a kernel $k_{\theta}\in\frakJ$ achieves universality on $\Omega$ when the associated $g_{\theta}$ is a homeomorphism from $\Omega$ to $g_{\theta}(\Omega)\subseteq\bbR^d$. However, in practical applications, $g_{\theta}$ is often implemented as neural network functions (with $\theta$ representing the trainable weights), making it impracticable to enforce invertibility for all $g_{\theta}$ in $\frakJ$. 

For any $g_{\theta}$ in $\frakJ$ that is not invertible, the corresponding kernel $k_{\theta}$ fails to be universal. As a consequence, the associated RKHS $\calK_{\theta}(\Omega)$ may not be dense in $C(\Omega)$, resulting in
\begin{align*}
	\mmd[{\calK_{\theta}}_{\infty}(\Omega),\bbP,\bbQ]
	\leq\mmd[C_{\infty}(\Omega),\bbP,\bbQ]=\TV(\bbP,\bbQ).
\end{align*}

Following \cref{mmd:uniform_norm}, for any $k_{\theta}\in\frakJ$, we have
\begin{align*}
	\mmd[\wtil{\calK}_{\theta}(\Omega),\bbP,\bbQ]
	\leq\mmd[{\calK_{\theta}}_{\infty}(\Omega),\bbP,\bbQ]
	\leq\TV(\bbP,\bbQ).
\end{align*}

To approximate $\TV(\bbP,\bbQ)$, we maximize the above lower-bound across the composition kernel candidates in $\frakJ$:
\begin{align*}
	\sup_{k_{\theta}\in\frakJ}\mmd[\wtil{\calK}_{\theta}(\Omega),\bbP,\bbQ].
\end{align*}

It is not difficult to identify the conditions under which maximizing $\mmd[\wtil{\calK}_{\theta}(\Omega),\bbP,\bbQ]$ reaches its upper limit.
\begin{Proposition}\label{props:var_mmd}
	Let $\frakJ$ be defined as \cref{eq:composition_ker} and $\Phi(\cdot)$ be the feature map associated with the fixed universal kernel $k$ on $\bbR^d$. For $\frakJ$, if there exists a $g_{\theta}$ and a point $x\in\Omega$ such that $g_{\theta}$ is a homeomorphism from $\Omega$ to $g_{\theta}(\Omega)\subseteq\bbR^d$ and $\Phi\circ g_{\theta}(x)\propto\E_{\bbP}[\Phi\circ g_{\theta}(X))]-\E_{\bbQ}[\Phi\circ g_{\theta}(X)]$, the following equality holds:
	\begin{align}\label{eq:var_mmd_equality}
		\sup_{k_{\theta}\in\frakJ}\mmd[\wtil{\calK}_{\theta}(\Omega),\bbP,\bbQ]
		=\TV(\bbP,\bbQ).
	\end{align}
\end{Proposition}

By maximizing $\mmd[\wtil{\calK}_{\theta}(\Omega),\bbP,\bbQ]$, we also obtain a composition kernel:
\begin{align}\label{eq:argsup_var_mmd}
	k_{\theta}=\argsup_{k_{\theta}\in\frakJ}\mmd[\wtil{\calK}_{\theta}(\Omega),\bbP,\bbQ].
\end{align}

It is crucial to note that $g_{\theta}$ being a homeomorphism in \cref{props:var_mmd} is one of the sufficient conditions for \cref{eq:var_mmd_equality} to hold. Therefore, even if \cref{eq:var_mmd_equality} holds, the optimally acquired composition kernel in \cref{eq:argsup_var_mmd} is not necessarily guaranteed to be universal.

However, our intention is to reuse the acquired composition kernel from \cref{eq:argsup_var_mmd} for the computation of $\chi^2(\bbP,\bbQ)$ in \cref{sect:ker_mmd_ub}. As a reminder from \cref{cor:CK}, $\mmd[\calK_2(\Omega,\nu),\bbP,\bbQ]$ is equivalent to $\chi^2(\bbP,\bbQ)$ if $\calK$ is a universal RKHS. Hence, our concern arises due to the potential lack of universality in the obtained kernel from \cref{eq:argsup_var_mmd} when applied to $\mmd[\calK_2(\Omega,\bbQ),\bbP,\bbQ]$. Interestingly, \cref{props:ker_ub_mmd} unveils that the kernel in \cref{eq:argsup_var_mmd} attains equivalence to a universal kernel for computing $\mmd[\calK_2(\Omega,\bbQ),\bbP,\bbQ]$ when \cref{eq:var_mmd_equality} holds.

\subsection{Approximation of $\mmd[\calK_2(\Omega,\bbQ),\bbP,\bbQ]$}
\label{sect:ker_mmd_ub}
In what follows, we derive the RKHS formula for $\mmd[\calK_2(\Omega,\bbQ),\bbP,\bbQ]$. 

To begin with, we introduce some basic concepts in functional analysis \citep{Con:B90}. For a given RKHS $\calK$, the operator norm of a linear operator $T:\calK\to\calK$ is defined to be $\norm*{T}_{\calK}=\sup_{\norm*{f}_{\calK}\leq 1}\norm*{Tf}_{\calK}$. The operator $T$ is invertible if and only if $Tf=0$ implies $f=0$. For $f,g\in\calK$, the tensor product $f\otimes g:\calK\to\calK$ is defined to be
\begin{align*}
	(f\otimes g)h\coloneqq f\ip*{g}{h}_{\calK},
\end{align*}
for any $h\in\calK$.

Next we proceed to characterize the class of witness functions $\calK_2(\Omega,\bbQ)=\left\{f\in\calK:\E_{\bbQ}[f^2(X)]\leq 1\right\}$.
Utilizing the definition of tensor product and reproducing property, for any $f,g\in\calK$, we have 
\begin{align*}
	\E_{\bbQ}[f(X)g(X)]
	=\E_{\bbQ}[\ip*{f}{\Phi(X)\otimes\Phi(X) g}_{\calK}].
\end{align*}
If \cref{asm:int_ker} is satisfied, the above function $\calK\times\calK\to\bbR$ is a bounded bilinear functional. By \citet[Theorem 4.3.13]{DebMik:B08}, there exists a unique operator $\Psi$ such that the quadratic form of this bilinear functional can be written as
\begin{align*}
	\E_{\bbQ}[f^2(X)]
	=\ip*{f}{\Psi f}_{\calK},
\end{align*}
where $\Psi=\E_{\bbQ}[\Phi(X)\otimes\Phi(X)]$.

It is verifiable that $\Psi$ is a self-adjoint and positive operator, and it is also a compact operator. To see this, 
letting $(e_i)_{i\geq 1}$ be an orthonormal basis of $\calK$, we have
\begin{align*}
	\Tr(\Psi)=\sum_{i\geq 1}\ip*{e_i}{\Psi e_i}_{\calK}
	&=\E_{\bbQ}[\ip*{\Phi(X)}{e_i}_{\calK}^2]\nn
	&\leq\E_{\bbQ}[\norm*{\Phi(X)}_{\calK}^2]=\E_{\bbQ}[k(X,X)]<\infty.
\end{align*}

Since $\Psi$ is compact, $\Psi$ is non-invertible. This limitation hinders the derivation of a closed-form formula for $\mmd[\calK_2(\Omega,\bbQ),\bbP,\bbQ]$. To address this, we introduce a perturbation to $\Psi$:
\begin{align*}
	\Psi_{\lambda}=\Psi+\lambda I,
\end{align*}
where $I$ denotes the identity operator and $\lambda>0$ is a regularization weight. Now, $\Psi_{\lambda}$ is an invertible and self-adjoint positive operator. Note every positive operator has a unique positive square root, which is also self-adjoint. The inverse and square root of $\Psi_{\lambda}$ are defined as
\begin{align*}
	&\Psi_{\lambda}^{-1}\Psi_{\lambda}f=f, \nn
	&\Psi_{\lambda}^{1/2}\Psi_{\lambda}^{1/2}f=\Psi_{\lambda}f,
\end{align*}
for any $f\in\calK$.

Subsequently, we introduce a proxy for $\calK_2(\Omega,\bbQ)$:
\begin{align*}
	\calK_2^{\lambda}(\Omega,\bbQ)
	\coloneqq\left\{f\in\calK(\Omega):\ip*{f}{\Phi_{\lambda}f}_{\calK}\leq 1\right\}.
\end{align*}

Note $\ip*{f}{\Phi_{\lambda}f}_{\calK}=\E_{\bbQ}[f^2(X)]+\lambda\norm*{f}^2_{\calK}$, and $\calK_2^{\lambda}(\Omega,\bbQ)$ can be interpreted as a regularized version of $\calK_2(\Omega,\bbQ)$.

We use $\mmd[\calK_2^{\lambda}(\Omega,\bbQ),\bbP,\bbQ]$ to approximate $\mmd[\calK_2(\Omega,\bbQ),\bbP,\bbQ]$, and we have
\begin{align}
	\mmd[\calK_2^{\lambda}(\Omega,\bbQ),\bbP,\bbQ]
	&=\sup_{f\in\calK} \braces*{\frac{\E_{\bbP}[f(X)]-\E_{\bbQ}[f(X)]}{\left(\E_{\bbQ}[f^2(X)]+\lambda\norm*{f}_{\calK}^2\right)^{1/2}}} \nn
	&=\sup_{f\in\calK}\frac{\ip*{f}{\mu_{\bbP}-\mu_{\bbQ}}_{\calK}}{\ip*{f}{\Psi_{\lambda}f}_{\calK}^{1/2}}\nn
	&=\sup_{f\in\calK}\frac{\ip*{\Psi_{\lambda}^{-1/2}\Psi_{\lambda}^{1/2}f}{\mu_{\bbP}-\mu_{\bbQ}}_{\calK}}{\ip*{\Psi_{\lambda}^{1/2}f}{\Psi_{\lambda}^{1/2}f}_{\calK}^{1/2}}\nn
	&=\sup_{f\in\calK}\frac{\ip*{\Psi_{\lambda}^{1/2}f}{\Psi_{\lambda}^{-1/2}(\mu_{\bbP}-\mu_{\bbQ})}_{\calK}}{\ip*{\Psi_{\lambda}^{1/2}f}{\Psi_{\lambda}^{1/2}f}_{\calK}^{1/2}}\nn
	&=\sup_{f\in\calK}\frac{\ip*{f}{\Psi_{\lambda}^{-1/2}(\mu_{\bbP}-\mu_{\bbQ})}_{\calK}}{\ip*{f}{f}_{\calK}^{1/2}}\nn
	&=\norm*{\Psi_{\lambda}^{-1/2}(\mu_{\bbP}-\mu_{\bbQ})}_{\calK}\nn
	&=\sqrt{\ip*{\mu_{\bbP}-\mu_{\bbQ}}{\Psi_{\lambda}^{-1}(\mu_{\bbP}-\mu_{\bbQ})}_{\calK}}. \label{eq:ub_kld}
\end{align}

\subsubsection{Approximation Error}\label{sect:chi2_appro_err}
It is of great interest to quantify the discrepancy between $\mmd[\calK_2(\Omega,\bbQ),\bbP,\bbQ]$ and its proxy $\mmd[\calK^{\lambda}_2(\Omega,\bbQ),\bbP,\bbQ]$.  

First of all, as $\lambda$ vanishes, $\mmd[\calK_2^{\lambda}(\Omega,\bbQ),\bbP,\bbQ]$ converges to $\mmd[\calK_2(\Omega,\bbQ),\bbP,\bbQ]$:
\begin{align}\label{eq:lim_chi2}
	\lim_{\lambda\to 0}\mmd[\calK_2^{\lambda}(\Omega,\bbQ),\bbP,\bbQ]=\mmd[\calK_2(\Omega,\bbQ),\bbP,\bbQ].
\end{align}
The derivation of equation \cref{eq:lim_chi2} is provided in \cref{proof:lim_chi2}.

Furthermore, $\mmd[\calK_2^{\lambda}(\Omega,\bbQ),\bbP,\bbQ]$ serves as a lower bound for $\mmd[\calK_2(\Omega,\bbQ),\bbP,\bbQ]$:
\begin{align*}
	\mmd^2[\calK^{\lambda}_2(\Omega,\bbQ),\bbP,\bbQ]
	&=\sup_{f\in\calK}\abs*{\frac{\ip*{f}{\mu_{\bbP}-\mu_{\bbQ}}_{\calK}^2}{\ip*{f}{\Psi_{\lambda}f}_{\calK}}}\nn
	&=\sup_{f\in\calK}\abs*{\frac{\ip*{f}{\mu_{\bbP}-\mu_{\bbQ}}_{\calK}^2}{\ip*{f}{\Psi f}_{\calK}+\lambda\norm*{f}_{\calK}^2}}\nn
	&\leq\dfrac{M_k}{M_k+\lambda}\mmd^2[\calK_2(\Omega,\bbQ),\bbP,\bbQ],
\end{align*}
where the last inequality arises from $\norm*{f}_{\calK}^2\geq\ofrac{M_k}\ip*{f}{\Psi f}_{\calK}$ due to $\opnorm{\Psi}\leq M_k$ (cf. \cref{appx:chi2_convas}).

It is clear that decreasing $\lambda$ (while keeping $M_k$ unchanged) theoretically diminishes the approximation error. However, in practice, a small $\lambda$ adversely affects the convergence rate of the estimation of $\mmd^2[\calK^{\lambda}_2(\Omega,\bbQ),\bbP,\bbQ]$. Choosing an appropriate $\lambda$ involves a trade-off between minimizing the approximation error and ensuring the reliability of the estimator, which will be elaborated in \cref{sect:chi2_est}.

\subsubsection{Applying Composition Kernel}
According to \cref{cor:CK}, a universal $\calK$ ensures $\mmd^2[\calK_2(\Omega,\bbQ),\bbP,\bbQ]=\chi^2(\bbP,\bbQ)$. As mentioned earlier, finding a universal kernel on the non-Euclidean sample space $\Omega$ is challenging. An alternative approach is to maximize $\mmd[\calK_2^{\lambda}(\Omega,\bbQ),\bbP,\bbQ]$ within a set of composition kernels, as demonstrated in \cref{sect:composition_ker}. However, compared to the maximization of $\mmd[\wtil{\calK}_{\theta}(\Omega),\bbP,\bbQ]$ in \cref{sect:composition_ker}, the computational complexity of the optimization in this case is significant, mainly due to the inverse of $\Psi_{\lambda}$. These computational challenges will be manifested in \cref{sect:ker_div_est}. 

Considering the computational challenges, a natural question arises: Can we reuse the composition kernel learned in \cref{sect:composition_ker} as a surrogate for a universal kernel, simplifying the computational problem? \cref{props:ker_ub_mmd} provides a positive answer.
\begin{Theorem}\label{props:ker_ub_mmd}
	Let $\frakJ$ be a set of kernels as defined in \cref{eq:composition_ker}, satisfying \cref{asm:int_ker}. Assume
	\begin{align*}
		\sup_{k_{\theta}\in\frakJ}\mmd[\wtil{\calK}_{\theta}(\Omega,\bbQ),\bbP,\bbQ]=\TV(\bbP,\bbQ).
	\end{align*}
	Let $\calK$ is the RKHS associated with kernel $k=\argsup_{k_{\theta}\in\frakJ}\mmd[\wtil{\calK}_{\theta}(\Omega,\bbQ),\bbP,\bbQ]$. Then we have
	\begin{align*}
		\lim_{\lambda\to 0}\mmd^2[\calK^{\lambda}_2(\Omega,\bbQ),\bbP,\bbQ]=\chi^2(\bbP,\bbQ).
	\end{align*}
\end{Theorem}

In essence, \cref{props:ker_ub_mmd} assets that the composition kernel obtained by maximizing the lower-bound \cref{eq:var_lb_mmd} is same as a universal kernel when applied to computation of \cref{eq:ub_kld}. This holds true even if the optimally acquired composition kernel is not universal. These theoretical insights enable us to overcome the difficulty of constructing a universal kernel on $\Omega$ in this case.

\section{Empirical Estimation}\label{sect:ker_div_est}
In this section, we derive sample-based estimates for $\sup_{k\in\frakJ}\mmd[\wtil{\calK}(\Omega),\bbP,\bbQ]$ and $\mmd[\calK^{\lambda}_2(\Omega,\bbQ),\bbP,\bbQ]$, respectively. These estimates serve as a practical approximation for $\TV(\bbP,\bbQ)$ and $\chi^2(\bbP,\bbQ)$. Furthermore, we present the statistical properties of these estimates, including the asymptotic analysis and the probability of error.

To start with, let's consider two sets of i.i.d. data samples $\{X_i\}_{i=1}^m$ and $\{Y_i\}_{i=1}^n$ drawn from distributions $\bbP$ and $\bbQ$, respectively. Additionally, we introduce i.i.d. random variables $X$ and $X'$, both following the distribution $\bbP$. Similarly, let $Y$ and $Y'$ be i.i.d. random variables following the distribution $\bbQ$.

\subsection{Approximation of TV Distance}
Firstly, we compute the estimate of $\mmd[\wtil{\calK}_{\theta}(\Omega),\bbP,\bbQ]$ based on its RKHS representation in \cref{eq:var_lb_mmd}. Then, this estimate is maximized over a set of composition kernels to approximate the TV distance, as illustrated in \cref{sect:composition_ker}.

We need a family of composition kernels $\frakJ=\{k_{\theta}(\cdot,\cdot):\theta\in\Theta\}$ as our search space for an optimal kernel, where $k_{\theta}(\cdot,\cdot)=k(g_{\theta}(\cdot),g_{\theta}(\cdot))$ and $k$ is a fixed kernel on Euclidean space. To ensure practicality, we let $g_{\theta}:\Omega\mapsto\bbR^d$ be a vector function parameterized by $\theta$, where $\theta=(\theta_1,\ldots,\theta_q)\in\bbR^q$, and $g_{\theta}$ is twice differentiable w.r.t. $\theta$, and the kernel $k$ on Euclidean space is also twice differentiable.

\subsubsection{Estimation of $\sup_{k_{\theta}\in\frakJ}\mmd[\wtil{\calK}_{\theta}(\Omega),\bbP,\bbQ]$}
The mean embeddings $\mu_{\bbP}$ and $\mu_{\bbQ}$ have sample-average estimates as
\begin{align*}
	&\what{\mu}_{\bbP}=\frac{1}{m}\sum_{i=1}^m\Phi(X_i)\ \text{and}\
	\what{\mu}_{\bbQ}=\frac{1}{n}\sum_{i=1}^n\Phi(Y_i).
\end{align*}
This yields the following estimate for the squared kernel mean distance \citep{GreBor:J06}:
\begin{align}
	\what{\mmd}^2[\wtil{\calK}_{\theta}(\Omega),\bbP,\bbQ]
	&=\ofrac{M_{k_{\theta}}}\norm*{\what{\mu}_{\bbP}-\what{\mu}_{\bbQ}}_{\calK_{\theta}}^2\nn
	&=\ofrac{M_{k_{\theta}}}\left(\frac{1}{m^2}\sum_{i=1}^m\sum_{j=1}^mk_{\theta}(X_i,X_j)+\frac{1}{n^2}\sum_{i=1}^n\sum_{i=1}^nk_{\theta}(Y_i,Y_j)\right.\nn
	&\left.\quad-\frac{2}{mn}\sum_{i=1}^m\sum_{i=1}^nk_{\theta}(X_i,Y_j)\right). \label{eq:ker_dist_est}
\end{align}

Note $\mmd^2[\wtil{\calK}_{\theta}(\Omega),\bbP,\bbQ]$ can be expanded as
\begin{align}\label{eq:ker_dist}
	\ofrac{M_{k_{\theta}}}\left(\E_{X,X'}k_{\theta}(X,X')+\E_{Y,Y'}k_{\theta}(Y,Y')-2\E_{X,Y}k_{\theta}(X,Y)\right).
\end{align}

This formulation enables a U-statistics estimate in linear computational time \citep{GreBor:J06}:
\begin{align*}
	\ofrac{M_{k_{\theta}}}\left(\frac{1}{m}\sum_{i=1}^{m-1}k_{\theta}(X_i,X_{i+1})+\frac{1}{n}\sum_{i=1}^{n-1}k_{\theta}(Y_i,Y_{i+1})-\frac{2}{n}\sum_{i=1}^nk_{\theta}(X_i,Y_i)\right).
\end{align*}

To approximate the TV distance, we maximize the estimate $\what{\mmd}[\wtil{\calK}_{\theta}(\Omega),\bbP,\bbQ]$ over $\theta\in\Theta$:
\begin{align}\label{tv_est}
	\sup_{\theta\in\Theta}\what{\mmd}[\wtil{\calK}_{\theta}(\Omega),\bbP,\bbQ].
\end{align}

The estimate in linear time facilitates an efficient application of the min-batch stochastic gradient when optimizing $\what{\mmd}[\wtil{\calK}_{\theta}(\Omega),\bbP,\bbQ]$.

The statistical properties of \cref{tv_est} are summarized in \cref{thm:tv_stats}.
\begin{Theorem}\label{thm:tv_stats}
	Denote the estimation error as
	\begin{align*}
		\Delta(\bbP,\bbQ)=
		\sup_{\theta\in\Theta}\what{\mmd}[\wtil{\calK}_{\theta}(\Omega),\bbP,\bbQ]
		-\sup_{\theta\in\Theta}\mmd[\wtil{\calK}_{\theta}(\Omega),\bbP,\bbQ].
	\end{align*}
	Let $\ell(\theta;x,y)=\dfrac{k_{\theta}(x,y)}{M_{k_{\theta}}}$. Assume $\ell(\theta;x,y)$ is Lipschitz w.r.t. $\theta$ uniformly over all $x,y\in\Omega$. That is, there exists a constant $C<\infty$ such that $\abs*{\ell(\theta;x,y)-\ell(\theta';x,y)}\leq C\norm*{\theta-\theta'}$ for all $x,y\in\Omega$. Then, as $m,n\to\infty$, we have
	\begin{enumerate}[(a)]
		\item $\Delta(\bbP,\bbQ)\convas 0$ if $\dfrac{m}{n}$ converges. 
		\item For any $t>0$, $\Delta(\bbP,\bbQ)$ achieves the following convergence rate:
		\begin{align*}
			&\P(\abs*{\Delta(\bbP,\bbQ)}\geq 3t+\dfrac{2}{m}+\dfrac{2}{n})\nn
			&\leq\exp\left(-\dfrac{m}{2}t^2\right)+\exp\left(-\dfrac{n}{2}t^2\right)+\exp\left(-\dfrac{mn}{m+n}t^2\right).
		\end{align*}
	\end{enumerate}
\end{Theorem}

\cref{thm:tv_stats} implies that we need balanced sample numbers for $\bbP$ and $\bbQ$ to maintain the consistency of the estimate.

\subsubsection{Learnable Composition Kernel}\label{sect:est_theta}
Through maximizing the estimate $\what{\mmd}[\wtil{\calK}_{\theta}(\Omega),\bbP,\bbQ]$, we obtain an estimate of the optimal kernel parameter $\what{\theta}\in\Theta$ that determines the optimal composition kernel. As the composition kernel associated with $\what{\theta}$ will be used for the computation of $\chi^2$-divergence, there is a keen interest in characterizing the statistical properties of the estimated parameters $\what{\theta}$ for the compositional kernel, as outlined below.

To facilitate our analysis, we denote the objective function of the optimization and its estimate w.r.t. the kernel parameter $\theta$ as
\begin{align*}
	&L(\theta)=\mmd[\wtil{\calK}_{\theta}(\Omega),\bbP,\bbQ],\nn
	&\what{L}(\theta)=\what{\mmd}[\wtil{\calK}_{\theta}(\Omega),\bbP,\bbQ].
\end{align*}
Additionally, we let the optimal kernel parameter and its estimate be
\begin{align*}
	&\theta^*=\argsup_{\theta\in\Theta}L(\theta),\nn
	&\what{\theta}=\argsup_{\theta\in\Theta}\what{L}(\theta).
\end{align*}

Furthermore, we denote $\nabla L(\theta)$ and $\nabla^2L(\theta)$ as the gradient and Hessian of $L(\theta)$ w.r.t. $\theta$, respectively. Recall the notation of operator norm $\opnorm{\Psi}$, i.e., $\opnorm{\Psi}=\inf\{c\geq 0, \norm*{\Psi f}_{\calK}\leq c\norm*{f}_{\calK}\}$.
\begin{Theorem}\label{thm:theta_stats}
	Assume $\what{\theta}$ and $\theta^*$ are critical points of $\what{L}(\theta)$ and $L(\theta)$, respectively. That is $\nabla\what{L}(\what{\theta})=0$ and $\nabla L(\theta^*)=0$.
	\begin{enumerate}[(a)]
		\item When $m=n$, we have
		\begin{align*}
			\sqrt{n}\left(\what{\theta}-\theta^*\right)
			\convd\N{0}{4\left(\nabla^2L(\theta^*)\right)^{-1}\bSigma\left(\nabla^2L(\theta^*)\right)^{-1}},
		\end{align*}
		where $\bSigma=\cov\left(\E_{X',Y'}[\nabla h(\theta^*;X,Y,X',Y')]\right)$, with the function $h$ defined as
		\begin{align*}
			h(\theta;X,Y,X',Y')
			=\odfrac{M_{k_{\theta}}}\left(k_{\theta}(X,X')-2k_{\theta}(X,Y')+k_{\theta}(Y,Y')\right).
		\end{align*}
		\item Let $\ell(\theta;x,y)=\dfrac{k_{\theta}(x,y)}{M_{k_{\theta}}}$. Assume $\norm*{\nabla^2\ell(\theta;x,y)}\geq\gamma$ and $\abs*{\ppfrac{\ell(\theta;x,y)}{\theta_i}}\leq\tau_i$ for all $\theta\in\Theta$ and $x,y\in\Omega$. Then for any $t>0$, we have
		\begin{align*}
			\P(\norm*{\what{\theta}-\theta^*}_2\geq\gamma^{-1}\sqrt{\sum_{i=1}^q\left(3t+\dfrac{4\tau_i}{n}\right)^2})
			\leq 3\sum_{i=1}^q\exp\left(-\dfrac{n}{2}\dfrac{t^2}{\tau_i^2}\right).
		\end{align*}
	\end{enumerate}
\end{Theorem}
Note the result presented in \cref{thm:theta_stats} (b) assumes $m=n$ for simplicity, while the derivation is conducted under the more general condition $m\neq n$. 

From \cref{thm:theta_stats} (b), it is evident that a larger $\gamma$ and smaller $\tau_i$ result in a faster convergence speed. To illustrate the impact of the design of $g_{\theta}$, consider the RBF kernel, i.e., $\ell(\theta;x,y)=\exp\left(-\norm*{g_{\theta}(x)-g_{\theta}(y)}_2^2\right)$ (since $M_{k_{\theta}}=1$ for any $\theta\in\Theta$). In this case, we have
\begin{align*}
	\nabla\ell(\theta;x,y)
	\leq\left(\ppfrac{g_{\theta}(x)}{\theta}-\ppfrac{g_{\theta}(y)}{\theta}\right)\left(g_{\theta}(x)-g_{\theta}(y)\right).
\end{align*}
Thus, limiting the scale of the output of $g_{\theta}(x)$ and enhancing the smoothness of $g_{\theta}$ allow us to achieve a small $\tau_i$.

\subsection{Approximation of $\chi^2$-Divergence}\label{sect:chi2_est}
In this section, we derive the estimate for $\mmd[\calK^{\lambda}_2(\Omega,\bbQ),\bbP,\bbQ]$ by leveraging its RKHS representation as presented in \cref{eq:ub_kld}. In addition, we explore the factors impacting the error probability.

\subsubsection{Estimation of $\mmd[\calK^{\lambda}_2(\Omega,\bbQ),\bbP,\bbQ]$}
The mean estimate of the linear operator $\Psi$ in \cref{eq:ub_kld} is given by:
\begin{align}\label{eq:LO_est}
	\what{\Psi}
	=\ofrac{n}\sum_{i=1}^n\Phi(Y_i)\otimes\Phi(Y_i).
\end{align}

Let $\bPhi=\left[\Phi(Y_1),\ldots,\Phi(Y_n)\right]$ be a row vector of feature maps, and define $\bPhi\T$ as a linear operator such that for $f\in\calK$,
\begin{align*}
	\bPhi\T f
	=\left[\ip*{f}{\Phi(Y_1)}_{\calK},\ldots,\ip*{f}{\Phi(Y_n)}_{\calK}\right]\T
	=\left[f(Y_1),\ldots,f(Y_n)\right]\T.
\end{align*}

Then, with a slight abuse of notation, we can rewrite \cref{eq:LO_est} as $\what{\Psi}=\odfrac{n}\bPhi\bPhi\T$, for which $\bPhi\bPhi\T f$ is the matrix product of $\bPhi$ and $\bPhi\T f$ for $f\in\calK$, interpreting $\what{\Psi}$ as an $n\times n$ matrix. Applying the Woodbury matrix identity, we have
\begin{align*}
	\what{\Psi}_{\lambda}^{-1}
	&=\inv{\lambda{I}+\what{\Psi}}\nn
	&=\inv{\lambda{I}+\ofrac{n}\bPhi\bPhi\T}\nn
	&=\ofrac{\lambda}\bI_n-\ofrac{\lambda}\bPhi\left(n\lambda\bI_n+\bPhi\T\bPhi\right)^{-1}\bPhi\T,
\end{align*}
where $\bI_n$ is the $n\times n$ identity matrix, and $\bPhi\T\bPhi$ denotes the application of $\bPhi\T$ on each column of $\bPhi$. It turns out that $\bPhi\T\bPhi=\bPhi_{\bbQ}$ is a matrix whose $(i,j)$-th entry is defined as $\subs{\bPhi_{\bbQ}}{i,j}=k(Y_i,Y_j)$ for $1\leq i,j\leq n$.

Now, substituting the estimates $\what{\mu}_{\bbP}$, $\what{\mu}_{\bbQ}$ and $\what{\Psi}_{\lambda}^{-1}$ for $\mu_{\bbP}$, $\mu_{\bbQ}$ and $\Psi_{\lambda}^{-1}$ in \cref{eq:ub_kld}, we obtain
\begin{align}\label{kld_ub_est}
	&\what{\mmd}^2[\calK^{\lambda}_2(\Omega,\bbQ),\bbP,\bbQ]\nn
	&=\ip*{\what{\mu}_{\bbP}-\what{\mu}_{\bbQ}}{\what{\Psi}_{\lambda}^{-1}(\what{\mu}_{\bbP}-\what{\mu}_{\bbQ})}_{\calK}\nn
	&=\ofrac{\lambda}\norm*{\what{\mu}_{\bbP}-\what{\mu}_{\bbQ}}_{\calK}^2-\ofrac{\lambda}\bw\T\left(n\lambda\bI_n+\bPhi_{\bbQ}\right)^{-1}\bw,
\end{align}
where $\bw\T=\left[w_1,\ldots,w_n\right]\in\bbR^n$ with the $i$-th element defined as
\begin{align*}
	w_i=\odfrac{m}\sum_{j=1}^mk(X_i,Y_j)-\odfrac{n}\sum_{j=1}^nk(Y_i,Y_j).
\end{align*} 

Comparing the expressions \cref{eq:ker_dist_est,kld_ub_est}, it is noted that \cref{kld_ub_est} involves an additional term that requires matrix inversion, significantly increasing  the computational complexity. Therefore, maximizing it over a set of composition kernels to locate the best kernel is not efficient. As aforementioned, the optimal strategy is to reuse the composition kernel obtained from \cref{sect:est_theta}. The validity of this approach is supported by \cref{props:ker_ub_mmd} in theory.

\subsubsection{Statistical Analysis}
The statistical properties of $\what{\mmd}^2[\calK^{\lambda}_2(\Omega,\bbQ),\bbP,\bbQ]$ are outlined in \cref{thm:chi2_conv}.
\begin{Theorem}\label{thm:chi2_conv}
	Denote the estimation error as
	\begin{align*}
		\Delta(\bbP,\bbQ)
		=\what{\mmd}^2[\calK^{\lambda}_2(\Omega,\bbQ),\bbP,\bbQ]-{\mmd}^2[\calK^{\lambda}_2(\Omega,\bbQ),\bbP,\bbQ].
	\end{align*}
	If \cref{asm:int_ker} is satisfied, as $m,n\to\infty$, we have
	\begin{enumerate}[(a)]
		\item $\Delta(\bbP,\bbQ)\convas 0$ if $\dfrac{n}{m}$ converges.
		\item When $m=n$, we have
		\begin{align*}
			\sqrt{n}\Delta(\bbP,\bbQ)\convd\N{0}{4\sigma^2},
		\end{align*}
		where $\sigma^2=\var\left(\ip*{\Phi(X)-\Phi(Y)}{\Psi_{\lambda}^{-1}(\mu_{\bbP}-\mu_{\bbQ})}_{\calK}\right)$.
		\item Denote $\alpha=\norm*{\Psi_{\lambda}^{-1}\mu}_{\calK}$ and $\beta=\norm*{\what{\Psi}_{\lambda}^{-1}\what{\mu}}_{\calK}$. Let $m=n$. For any $t_1>0$ ,$t_2>0$ and $t_3>0$, we have
		\begin{align*}
			&\P(\abs*{\Delta(\bbP,\bbQ)}\geq t_1+t_2+t_3)\nn
			&\leq\exp\left(-\dfrac{2nt_1^2}{M_k^2\alpha^2\beta^2}\right)+2\exp\left(-\dfrac{nt_2^2}{2M_k\beta^2}\right)+2\exp\left(-\dfrac{nt_3^2}{2M_k\alpha^2}\right).
		\end{align*}
	\end{enumerate}
\end{Theorem}

Note the result presented in \cref{thm:chi2_conv} (c) assumes $m=n$ for simplicity, while the derivation is conducted under the more general condition $m\neq n$. 

In the ensuing discussion, we delve into the implications of the error probability presented in \cref{thm:chi2_conv} (c). 
To facilitate our analysis, we set the error probability to be a constant by letting $t_1=\dfrac{M_k\alpha\beta}{\sqrt{n}}$, $t_2=\dfrac{\sqrt{M_k}\beta}{\sqrt{n}}$ and $t_3=\dfrac{\sqrt{M_k}\alpha}{\sqrt{n}}$. With high probability, we can assert the following bound on the estimation error:
\begin{align}\label{chi2_err_lb}
	\abs*{\Delta(\bbP,\bbQ)}
	\leq\dfrac{M_k\alpha\beta}{\sqrt{n}}+\dfrac{\sqrt{M_k}\alpha}{\sqrt{n}}+\dfrac{\sqrt{M_k}\beta}{\sqrt{n}}.
\end{align}

Several observations can be made by examining \cref{chi2_err_lb}:
\begin{itemize}
	\item In considering the worst-case scenario for the error bound \cref{chi2_err_lb} achieved by substituting the following bounds of $\alpha$ and $\beta$ into \cref{chi2_err_lb}: $\alpha\leq\lambda^{-1}\sqrt{M_k}$ and $\beta\leq\lambda^{-1}\sqrt{M_k}$ (cf. \cref{appx:chi2_convas}), we get
	\begin{align*}
		\abs*{\Delta(\bbP,\bbQ)}
		\leq\dfrac{M_k^2\lambda^{-2}}{\sqrt{n}}+\dfrac{M_k\lambda^{-1}}{\sqrt{n}}+\dfrac{M_k\lambda^{-1}}{\sqrt{n}}.
	\end{align*} 
	
	It is evident that decreasing $M_k$ or increasing $\lambda$ diminishes the estimation error. However, this reduction comes with a concomitant enlargement of the approximation error, as illustrated by the approximation discrepancy presented in \cref{sect:chi2_appro_err}:
	\begin{align} \label{chi2:appr_err}
		\mmd^2[\calK^{\lambda}_2(\Omega,\bbQ),\bbP,\bbQ]
		\leq\dfrac{M_k}{M_k+\lambda}\mmd^2[\calK_2(\Omega,\bbQ),\bbP,\bbQ].
	\end{align}
	
	For instance, when letting $\lambda=M_k=1$, the estimation error decreases at the rate of $\sqrt{n}$ as $n$ increases. However, it is important to note that $\mmd^2[\calK^{\lambda}_2(\Omega,\bbQ),\bbP,\bbQ]$ only lower-bounds half the value of $\chi^2$-divergence. To conclude, for a fixed number of samples, there is an inherent trade-off between the reliability of the estimator and the approximation error.
	
	\item On the other hand, it is noteworthy that $\alpha$ is a deterministic variable, and its magnitude is positively correlated with ${\mmd}[\calK^{\lambda}_2(\Omega,\bbQ),\bbP,\bbQ]$, as demonstrated below. Let $(e_p,\lambda_p)_{p\geq 1}$ denote the eigenvalues, arranged in non-increasing order, and the corresponding eigenfunctions of $\Psi$ on RKHS $\calK$, where $\lambda_p\geq 0$ since $\Psi$ is a positive operator. Consequently, we have:
	\begin{align*}
		&\alpha^2=\norm*{\Psi_{\lambda}^{-1}\mu}_{\calK}^2
		=\sum_{p\geq 1}(\lambda_p+\lambda)^{-2}\ip*{e_p}{\mu}_{\calK}^2,\ \text{and} \nn
		&{\mmd}^2[\calK^{\lambda}_2(\Omega,\bbQ),\bbP,\bbQ]
		=\sum_{p\geq 1}(\lambda_p+\lambda)^{-1}\ip*{e_p}{\mu}_{\calK}^2.
	\end{align*}
	
	It is clear from the above expressions that when ${\mmd}^2[\calK^{\lambda}_2(\Omega,\bbQ),\bbP,\bbQ]$ is small, $\ip*{e_p}{\mu}_{\calK}^2$ tends to be small for those small $\lambda_p$, simultaneously leading to a small $\alpha$. This implies that the estimation error bound in \cref{thm:chi2_conv} (c) adapts to the magnitude of the true value, allowing us to ignore the impact of $\alpha$ on the estimation error.
	
	As for $\beta=\norm*{\what{\Psi}_{\lambda}^{-1}\what{\mu}}_{\calK}$, which is tied to an estimation, its reduction contributes to minimizing the estimation error. Therefore, an ideal kernel should render $\beta$ small. Motivated by this, we can seek a kernel by maximizing
	\begin{align*}
		\norm*{\what{\Psi}_{\lambda}^{1/2}\what{\mu}}_{\calK}^2
		=\odfrac{n}\sum_{i=1}^n\left(\odfrac{m}\sum_{j=1}^mk(X_i,X_j)-\odfrac{n}\sum_{j=1}^nk(X_i,Y_j)\right)^2.
	\end{align*}
	
	Alternatively, this quantity can also be employed as a regularization term in the maximization of \cref{eq:ker_dist_est}.
\end{itemize}

\section{Numerical Experiments}\label{sect:expt}
In this section, we evaluate the effectiveness of our proposed estimators \cref{tv_est,kld_ub_est} designed for approximating the TV distance and $\chi^2$-divergence on non-Euclidean space, by applying them to the UTK face dataset \citep{ZhaSonQi:C17}.

The UTK face dataset \citep{ZhaSonQi:C17} consists of over $20,000$ face images annotated with gender. In our experiment, we vary the discrepancy between two image domains and estimate the domain discrepancy. In the beginning, each of the two domains contains 1200 female images, sampled from the UTK dataset. We gradually replace the female images in one domain with male images, intentionally increasing the discrepancy between the two domains. At each step, we perform the estimation $20$ times. The results are depicted in \cref{fig:TV,fig:ChiSquare,fig:MINE}, where the horizontal axis representing the percentage of female images that are replaced by male images. A larger replacement ratio corresponds to a greater discrepancy.

We adopt the Gaussian kernel as the universal kernel on Euclidean space. The neural network function $g_{\theta}:\Omega\to\bbR^d$ employed for the composition kernel is selected as resnet18 \citep{HeZhaRena:C16}. The estimate of the TV distance is obtained by maximizing the objective function in \cref{eq:ker_dist_est} w.r.t. $g_{\theta}$. As previously mentioned, we reuse the composition kernel associated with the optimal $g_{\theta}$ for computing the estimate in \cref{kld_ub_est} for approximating the $\chi^2$-divergence.

We compare our estimators with two other methods: the plain kernel method (direct application of the Gaussian kernel) and the parametric method called mutual information neural estimator (MINE) \citep{BelBarRaj:C18}. As depicted in \cref{fig:TV,fig:ChiSquare}, the plain kernel method fails to capture the increasing trend of the discrepancy between two image domains, while our semi-parametric estimators exhibit a smooth performance. This is because the image data are not inherently compatible with the Euclidean distance, whereas our learnable function $g_{\btheta}$ for composition kernel implicitly captures the underlying topology of the input data. The performance of MINE, shown in \cref{fig:MINE}, suffers from high uncertainty and includes some outliers that are not depicted in the figure.

\begin{figure}[!htb]
	\centering
	\begin{minipage}{.32\textwidth}
		\centering
		\includegraphics[width=\textwidth]{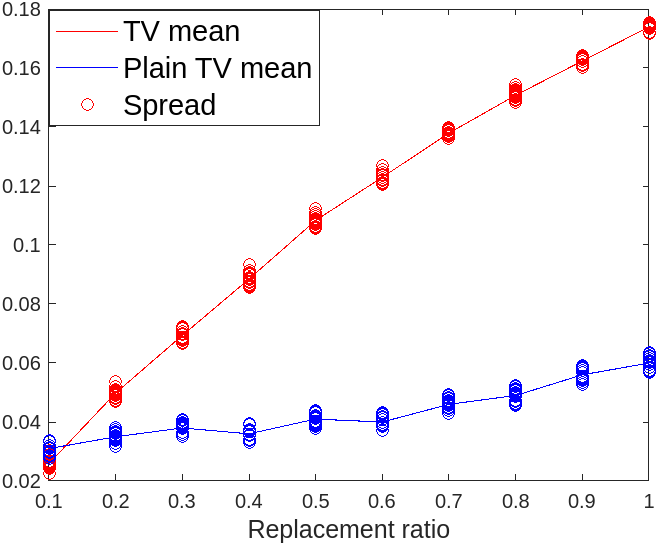}
		\caption{Total variation.}
		\label{fig:TV}
	\end{minipage}
	\begin{minipage}{.32\textwidth}
		\centering
		\includegraphics[width=\textwidth]{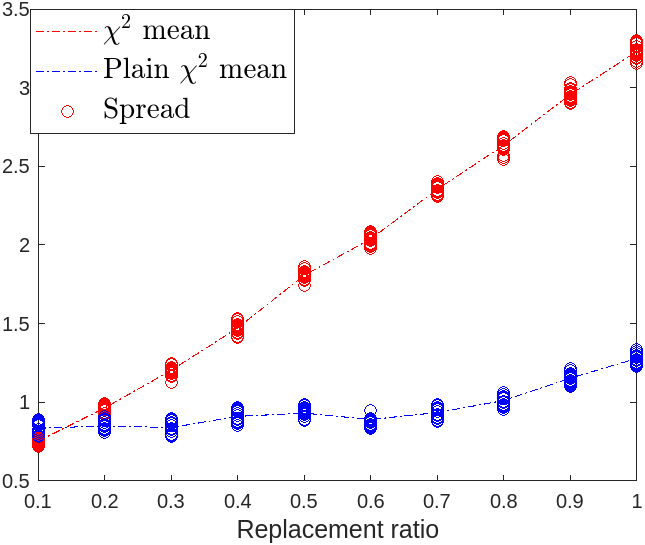}
		\caption{$\chi^2$-divergence.}
		\label{fig:ChiSquare}
	\end{minipage}
	\begin{minipage}{.32\textwidth}
		\centering
		\includegraphics[width=\textwidth]{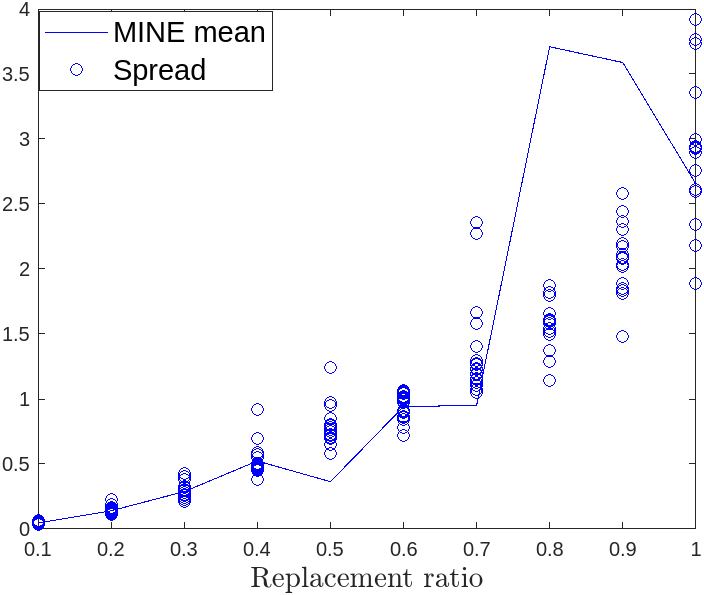}
		\caption{MINE estimator.}
		\label{fig:MINE}
	\end{minipage}
\end{figure}

In addition to the plain kernel method and the MINE method discussed in the main article, we use the Wasserstein distance (WD) \citep{ArjChiBot:C17} and the $k$-nearest neighbor ($k$-NN) method \citep{WanKulVer:J09} to measure the divergence between two face image domains. For our experiments, we selected a value of $k=1000$ for the $k$-NN method. To ensure the robustness and reliability of the results, each experiment was independently conducted 20 times.

As depicted in \cref{fig:wgan}, the divergence measured by WD exhibits volatility, similar to the MINE method. This behavior could be attributed to the heuristic implementation of WD. The neural network function used for WD estimation requires constant Lipschitz continuity. In practice, we truncate the weights of the neural network to approximate a continuous function with a Lipschitz constant. 

The $k$-NN method is typically employed with continuous distributions in Euclidean space. Hence, it is not surprising to observe its limited effectiveness when applied directly to the raw image domain, as shown in \cref{fig:kNN}. However, when we project raw images onto a Euclidean space using our proposed method and subsequently apply $k$-NN to the feature data in the Euclidean space, it performs well. This demonstrates the capability of our approach to learn a reliable transformation function from the non-Euclidean manifold to the Euclidean space.

\begin{figure}[!htb]
	\centering
	\begin{minipage}{.45\textwidth}
		\centering
		\includegraphics[width=\textwidth]{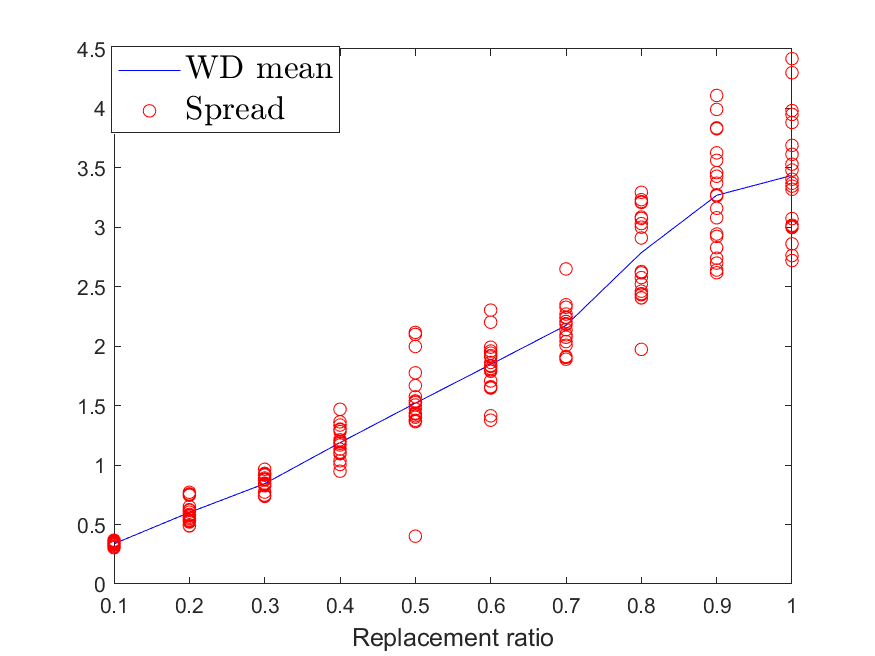}
		\caption{Wasserstein Distance.}
		\label{fig:wgan}
	\end{minipage}
	\begin{minipage}{.45\textwidth}
		\centering
		\includegraphics[width=\textwidth]{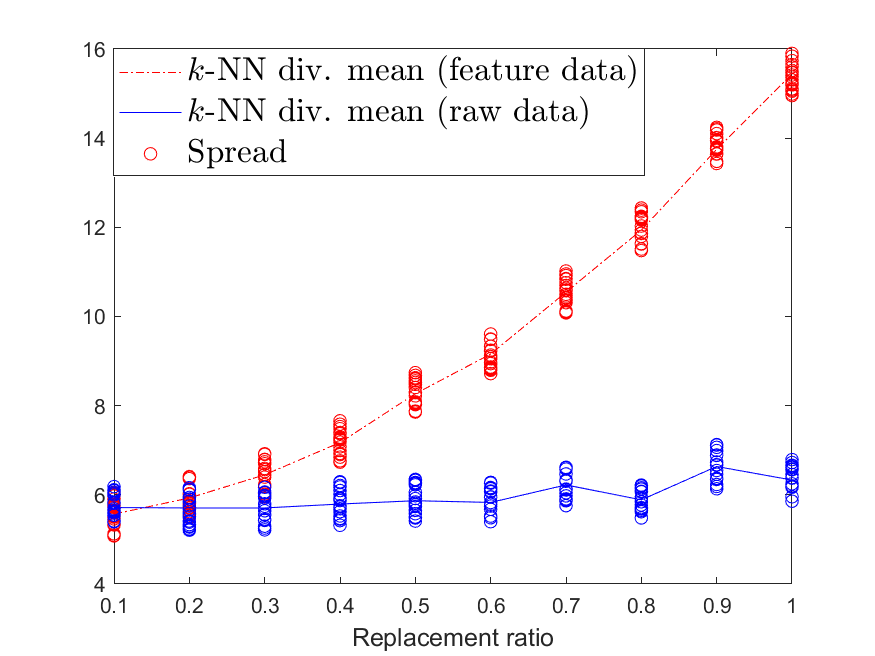}
		\caption{k-NN Divergence.}
		\label{fig:kNN}
	\end{minipage}
\end{figure}

\section{Conclusion}
\label{sect:concl}

We have introduced two IPMs for computing the TV distance and $\chi^2$-divergence on topological spaces that are homeomorphic to Euclidean subsets. In our approach, we have derived closed-form solutions within RKHSs for the proposed IMPS. This involves combining a universal kernel on Euclidean space with a function from topological spaces to Euclidean space. We have performed the statistical analysis for the empirical estimators of the proposed estimators.    

The statistical-computational trade-offs associated with the proposed estimators are left open.


\appendix
\section{Proof of \cref{thm:KLDbounds}}\label[Appendix]{appx:ipm_tv_chi2}
The total variation (TV) distance between probability measures $\bbP$ and $\bbQ$ is defined as 
\begin{align*}
	\TV(\bbP,\bbQ)\coloneqq 2\sup_{A\in\calB(\Omega)}\abs{\bbP(A)-\bbQ(A)},
\end{align*}
where $\calB(\Omega)$ is the Borel $\sigma$-algebra associated with $\Omega$.

We prove $\mmd[C_{\infty}{(\Omega)},\bbP,\bbQ]=\TV(\bbP,\bbQ)$. Denote $\calG=\left\{f : \supnorm{f}\leq 1\right\}$ as a set of bounded functions measurable w.r.t. $\bbQ$. Letting $\calF=\calG$ in $\mmd(\calF,\bbP,\bbQ)$, we have
\begin{align*}
	\mmd(\calG,\bbP,\bbQ)
	&=\sup_{f\in \calG}\left\{\int_{\Omega}f\left(\ddfrac{\bbP}{\bbQ}-1\right)\ud\bbQ\right\}\\
	&=\int_{\Omega}\abs*{\ddfrac{\bbP}{\bbQ}-1}\ud\bbQ\\
	&=\int_{\Omega}\left(\ddfrac{\bbP}{\bbQ}-1\right)^+ \ud\bbQ + \int_{\Omega}\left(\ddfrac{\bbP}{\bbQ}-1\right)^- \ud\bbQ\\	
	&=\TV(\bbP,\bbQ),
\end{align*}
where $X^+=\max\{X,0\}$, $X^-=-\min\{X,0\}$. Note the supremum in $\mmd(\calG,\bbP,\bbQ)$ is attained at 
\begin{align*}
	f=
	\begin{cases}
		1,\ &\text{for}\ \ddfrac{\bbP}{\bbQ}\geq 1, \\
		-1,\ &\text{otherwise}.
	\end{cases}
\end{align*}
It is easy to verify that there exists a sequence of functions in $C_{\infty}(\Omega)$ that can uniformly approximate such $f$. As a consequence, we have
\begin{align*}
	\mmd[C_{\infty}{(\Omega)},\bbP,\bbQ]=\mmd(\calG,\bbP,\bbQ)=\TV(\bbP,\bbQ).
\end{align*}

Now we prove $\chi^2(\bbP,\bbQ)=\mmd^2[C_2(\Omega,\bbQ),\bbP,\bbQ]$. We assume that both the probability measures $\bbP$ and $\bbQ$ are absolutely continuous w.r.t. a measure $\nu$ associated to the norm in \cref{eq:cls_l2_norm} and $\ddfrac{\bbP}{\nu}$ and $\ddfrac{\bbP}{\nu}$ are continuous on $\Omega$. From the Cauchy-Schwarz inequality, we have
\begin{align*}
	\mmd[C_2(\Omega,\nu),\bbP,\bbQ]
	&=\sup_{f\in C_2(\Omega,\nu)}\int_{\Omega} (f \ud\bbP - f \ud\bbQ) \nn
	&=\sup_{f\in C_2(\Omega,\nu)}\int_{\Omega}f\left(\ddfrac{\bbP}{\nu}-\ddfrac{\bbQ}{\nu}\right) \ud\nu \nn
	&=\norm*{\ddfrac{\bbP}{\nu}-\ddfrac{\bbQ}{\nu}}_{\nu}.
\end{align*}
Letting $\nu=\bbQ$, we have
\begin{align*}
	\mmd^2[C_2(\Omega,\bbQ),\bbP,\bbQ]
	&=\norm*{\ddfrac{\bbP}{\bbQ}-1}^2_{\bbQ}\nn
	&=\int_{\Omega}\left(\ddfrac{\bbP}{\bbQ}-1\right)^2 \ud\bbQ \nn
	&=\int_{\Omega}\ddfrac{\bbP}{\bbQ} \ud\bbP-1
	=\chi^2(\bbP,\bbQ).
\end{align*}
The proof is complete.

The TV distance and $\chi^2$-divergence can be linked to the KL divergence \citep{BreHub:J78,Ver:C14} as: 
\begin{align*}
	&\TV(\bbP,\bbQ)\leq 2\sqrt{1-\exp(-D_{\mathrm{KL}}(\bbP\ \|\ \bbQ))},\nn
	&\chi^2(\bbP,\bbQ)\geq\exp(D_{\mathrm{KL}}(\bbP\ \|\ \bbQ)).
\end{align*}

\section{Proofs of \cref{cor:CK}}\label[Appendix]{appx:Proofs_equiv}

We need a preliminary lemma.
\begin{Lemma}\label{lemma:preserve_dense}
	If $\calF'$ is dense in $\calF$, we have $\mmd[\calF',\bbP,\bbQ]=\mmd[\calF,\bbP,\bbQ]$.
\end{Lemma}
\begin{proof}
Let $M=\mmd[\calF,\bbP,\bbQ]$.
For any $\epsilon>0$, there exists $f\in\calF$ such that
\begin{align*}
	\abs{\E_{\bbP}[f(X)]-\E_{\bbQ}[f(X)]-M}\leq\epsilon/3.
\end{align*}
In addition, there exists $g\in\calF'$ such that $\norm*{g-f}_{\infty}\leq\epsilon/3$.
We then have
\begin{align*}
	&\abs{\E_{\bbP}[g(X)]-\E_{\bbQ}[g(X)]-M} \\
	&\leq\abs{\E_{\bbP}[g(X)-f(X)]}+\abs{\E_{\bbQ}[g(X)-f(X)]}+\abs{\E_{\bbP}[f(X)]-\E_{\bbQ}[f(X)]-M}
	\leq\epsilon.
\end{align*}
Since $\epsilon$ is arbitrary, the proof of \cref{lemma:preserve_dense} is complete.
\end{proof}

Now it suffices to show (1) $\calK_{\infty}(\Omega)$ is dense in $C_{\infty}(\Omega)$, and (2) $\calK_2(\Omega,\nu)$ is dense in $C_2(\Omega,\nu)$.

(1) Let $f\in C_{\infty}(\Omega)$. For any $\epsilon \in (0,1)$, since $(1-\epsilon/2)f\in C(\Omega)$ and $\calK(\Omega)$ is dense in $C(\Omega)$, there exists $g\in\calK(\Omega)$, such that
\begin{align*}
	\supnorm{g-(1-\epsilon/2)f}\leq\epsilon/2.
\end{align*}
Then, by the triangle inequality, we have 
\begin{align*}
	\supnorm{g}
	\leq\supnorm{(1-\epsilon/2)f}+\supnorm{g-(1-\epsilon/2)f)}
	\leq 1,
\end{align*}
which implies $g\in\calK_{\infty}(\Omega)$. Finally, we have
\begin{align*}
	\supnorm{g-f}
	\leq\supnorm{g-(1-\epsilon/2)f}+\supnorm{f-(1-\epsilon/2)f}
	\leq\epsilon.
\end{align*}

(2)	Let $f\in C_2(\Omega,\nu)$. For any $\epsilon\in(0,1)$, since $(1-\epsilon/2)f\in C(\Omega)$, there exists $g\in\calK(\Omega)$, such that
\begin{align*}
	\supnorm{g-(1-\epsilon/2)f}
	\leq\epsilon/(2\nu(\Omega)).
\end{align*}
Then, we have
\begin{align*}
	\norm*{g-(1-\epsilon/2)f}_{\nu}
	\leq\nu(\Omega)\supnorm{g-(1-\epsilon/2)f}
	\leq\epsilon/2,
\end{align*}
Moreover, 
\begin{align*}
	\norm*{g}_{\nu}
	\leq\norm*{(1-\epsilon/2)f}_{\nu}+\norm*{g-(1-\epsilon/2)f}_{\nu}
	\leq 1,
\end{align*}
which implies $g\in\calK_2(\Omega,\nu)$. Finally, we have
\begin{align*}
	\norm*{g-f}_{\nu}
	\leq\norm*{g-(1-\epsilon/2)f}_{\nu}+\norm*{f-(1-\epsilon/2)f}_{\nu}
	\leq\epsilon.
\end{align*}

\section{Proof of \cref{props:ker_composition}}\label[Appendix]{appx:ker_comp}
Since $g:\Omega\to\Omega'$ is bijective and continuous and $\Omega$ is compact, $g^{-1}:\Omega'\to\Omega$ is continuous. For any $f\in C(\Omega)$, $f\circ g^{-1}$ is also a continuous function, leading to $f\circ g^{-1}\in C(\Omega')$. Therefore, for every $\epsilon> 0$, there exists an element $p\in\calK(\Omega')$ such that $\supnorm{p-f\circ g^{-1}}<\epsilon$. We have
\begin{align*}
	\supnorm{p\circ g-f}
	&=\max_{x\in\Omega}\abs{p\circ g(x)-f(x)} \nn
	&\leq\max_{x\in\Omega'}\abs{p(x)-f(g^{-1}(x))} \nn
	&=\supnorm{p-f\circ g^{-1}}<\epsilon,
\end{align*}
and the proof is complete.

\section{Achieving Equality for \cref{mmd:uniform_norm}}\label[Appendix]{appx:supnorm_equality}
Suppose $f'\in\calK$ is an optimal solution for computing $\mmd[\calK_{\infty}(\Omega),\bbP,\bbQ]$. Then we have
\begin{align*}
	\mmd[\calK_{\infty}(\Omega),\bbP,\bbQ]
	&=\sup_{f\in\calK} \braces*{\frac{\E_{\bbP}[f(X)]-\E_{\bbQ}[f(X)]}{\supnorm{f}}} \nn
	&=\frac{\ip*{f'}{\mu_{\bbP}-\mu_{\bbQ}}_{\calK}}{\supnorm{f'}}.
\end{align*}

From the Cauchy-Schwarz relaxation in \cref{eq:max_norm_ub}, it is known that if there is a $x\in\Omega$ such that $\Phi(x)\propto f'$, we have
\begin{align*}
	\supnorm{f'}=\sqrt{M_k}\norm{f'}_{\calK}.
\end{align*}

Then, we obtain
\begin{align*}
	\mmd[\calK_{\infty}(\Omega),\bbP,\bbQ]
	=\ofrac{\sqrt{M_k}}\frac{\ip*{f'}{\mu_{\bbP}-\mu_{\bbQ}}_{\calK}}{\norm{f'}_{\calK}},
\end{align*}
where $f'$ must equal $\mu_{\bbP}-\mu_{\bbQ}$, otherwise, it contradicts the supremum.

\section{Proof of Equation \cref{eq:lim_chi2}}\label[Appendix]{proof:lim_chi2}
\begin{align*}
	&\abs*{\mmd[\calK_2(\Omega,\bbQ),\bbP,\bbQ]-\mmd[\calK_2^{\lambda}(\Omega,\bbQ),\bbP,\bbQ]}\nn
	&=\abs*{\sup_{f\in\calK} \braces*{\frac{\E_{\bbP}[f(X)]-\E_{\bbQ}[f(X)]}{\E_{\bbQ}[f^2(X)]^{1/2}}}-\sup_{f\in\calK} \braces*{\frac{\E_{\bbP}[f(X)]-\E_{\bbQ}[f(X)]}{\left(\E_{\bbQ}[f^2(X)]+\lambda\norm*{f}^2_{\calK}\right)^{1/2}}}}\nn
	&\leq\sup_{f\in\calK}\abs*{\frac{\E_{\bbP}[f(X)]-\E_{\bbQ}[f(X)]}{\E_{\bbQ}[f^2(X)]^{1/2}}-\frac{\E_{\bbP}[f(X)]-\E_{\bbQ}[f(X)]}{\left(\E_{\bbQ}[f^2(X)]+\lambda\norm*{f}^2_{\calK}\right)^{1/2}}}\nn
	&=\sup_{\norm*{f}_{\calK}=1}\abs*{\frac{\E_{\bbP}[f(X)]-\E_{\bbQ}[f(X)]}{\E_{\bbQ}[f^2(X)]^{1/2}}-\frac{\E_{\bbP}[f(X)]-\E_{\bbQ}[f(X)]}{\left(\E_{\bbQ}[f^2(X)]+\lambda\norm*{f}^2_{\calK}\right)^{1/2}}}\nn
	&=\sup_{\norm*{f}_{\calK}=1}\abs*{\frac{\E_{\bbP}[f(X)]-\E_{\bbQ}[f(X)]}{\E_{\bbQ}[f^2(X)]^{1/2}}-\frac{\E_{\bbP}[f(X)]-\E_{\bbQ}[f(X)]}{\left(\E_{\bbQ}[f^2(X)]+\lambda\right)^{1/2}}}.
\end{align*}
The proof is complete by taking the limit at both sides as $\lambda\to 0$.

\section{Proof of \cref{props:ker_ub_mmd}}
Let $\calK$ be the RKHS associated with the kernel $k=\argsup_{k\in\frakJ}\mmd[\wtil{\calK}(\Omega,\bbQ),\bbP,\bbQ]$. If $\calK$ is universal, the claim is obvious by combining \cref{thm:KLDbounds,cor:CK,eq:lim_chi2}.

Suppose $\calK$ is not universal. Let $\calH$ be some universal RKHS. It suffices to show $\mmd[{\calK}_2^{\lambda}(\Omega,\bbQ),\bbP,\bbQ]=\mmd[{\calH}_2^{\lambda}(\Omega,\bbQ),\bbP,\bbQ]$ if $\mmd[\wtil{\calK}(\Omega),\bbP,\bbQ]=\mmd[\wtil{\calH}(\Omega),\bbP,\bbQ]$. 

Since $\calH_2(\Omega,\bbQ)$ is dense in $C_2(\Omega,\bbQ)$, we can always find a subspace $\calH'\subseteq\calH$ such that
\begin{align*}
	\mmd[\calK_2^{\lambda}(\Omega,\bbQ),\bbP,\bbQ]=\mmd[{\calH'}_2^{\lambda}(\Omega,\bbQ),\bbP,\bbQ].
\end{align*}
Therefore, in terms of computing $\mmd[\calK_2^{\lambda}(\Omega,\bbQ),\bbP,\bbQ]$, it is without loss of generality to assume $\calK\subseteq\calH$.  Denote $\mu_{\bbP}'$ and $\mu_{\bbQ}'$ as the kernel embeddings of $\bbP$ and $\bbQ$ w.r.t. $\calH$, and let $\Phi'(X)$ be the kernel feature mapping of $X$ w.r.t. $\calH$. We note
\begin{align*}
	\mmd[{\calH}_2^{\lambda}(\Omega,\bbQ),\bbP,\bbQ]
	=\sup_{f\in\calH}\frac{\ip*{f}{\mu_{\bbP}'-\mu_{\bbQ}'}_{\calH}}{\ip*{f}{\Psi_{\lambda}f}_{\calH}^{1/2}},
\end{align*}
which attains the maximum at $f\propto\Psi_{\lambda}^{-1/2}(\mu_{\bbP}'-\mu_{\bbQ}')$.

To show $\mmd[{\calK}_2^{\lambda}(\Omega,\bbQ),\bbP,\bbQ]=\mmd[{\calH}_2^{\lambda}(\Omega,\bbQ),\bbP,\bbQ]$, it remains to be proved that $\Psi_{\lambda}^{-1/2}(\mu_{\bbP}'-\mu_{\bbQ}')\in\calK$. 

We note
\begin{align*}
	\mmd[\wtil{\calH}(\Omega),\bbP,\bbQ]
	=\sup_{g\in\calH} \braces*{\ofrac{\sqrt{M_k}}\frac{\ip*{g}{\mu_{\bbP}-\mu_{\bbQ}}_{\calH}}{\norm*{g}_{\calH}}},
\end{align*}
which attains the maximum at $g\propto\mu_{\bbP}'-\mu_{\bbQ}'$. Therefore, under the condition $\mmd[\wtil{\calH}(\Omega),\bbP,\bbQ]=\mmd[\wtil{\calK}(\Omega),\bbP,\bbQ]$, we must have $\mu_{\bbP}'-\mu_{\bbQ}'\in\calK$. This implies $\mu_{\bbP}'\in\calK$ and $\mu_{\bbQ}'\in\calK$. 
For any $h\in\calK$, we have
\begin{align*}
	\Psi_{\lambda}h
	=\E_{\bbQ}[\Phi'(X)\otimes\Phi'(X)]h+\lambda h
	=\E_{\bbQ}[\Phi'(X)h(X)]+\lambda h\in\calK.
\end{align*}
This is because $\mu_{\bbQ}'=\E_{\bbQ}[\Phi'(X)]\in\calK$. Since $\Psi_{\lambda}$ is self-adjoint, $\Psi_{\lambda}^{-1/2}h\in\calK$ for $h\in\calK$. As a result, $\Psi_{\lambda}^{-1/2}(\mu_{\bbP}'-\mu_{\bbQ}')\in\calK$. The proof is complete.

\section{Proof of \cref{thm:tv_stats} (a)}\label[Appendix]{appx:tv_convas}
Firstly, we present \cref{appx:uniform_convas} for later use, whose proof is provided at the end of this section.
\begin{Lemma}\label{appx:uniform_convas}
	Suppose $T_n(\theta)\convas T(\theta)$ for any $\theta\in\Theta$ where $\Theta$ is a compact space. If $T_n(\theta)$ and $T(\theta)$ are Lipschitz w.r.t. $\theta$, e.g., there exists $K>0$ such that $\abs*{T(\theta)-T(\theta')}\leq K\norm*{\theta-\theta'}_2$, we have
	\begin{align*}
		\sup_{\theta\in\Theta}T_n(\theta)\convas\sup_{\theta\in\Theta}T(\theta).
	\end{align*}
\end{Lemma}

Now we return to the main proof. Because by assumption $\dfrac{k_{\theta}(x,y)}{M_{k_{\theta}}}$ is Lipschitz w.r.t. $\theta$ uniformly over all $x,y\in\Omega$, it can be easily verified that $\what{\mmd}^2[\wtil{\calK}_{\theta}(\Omega),\bbP,\bbQ]$ and $\mmd^2[\wtil{\calK}_{\theta}(\Omega),\bbP,\bbQ]$ are also Lipschitz.
By \cref{appx:uniform_convas}, it suffices to prove $\what{\mmd}^2[\wtil{\calK}_{\theta}(\Omega),\bbP,\bbQ]\convas\mmd^2[\wtil{\calK}_{\theta}(\Omega),\bbP,\bbQ]$.

The kernel $k_{\theta}$ is associated with a Hilbert–Schmidt integral operator $S$, defined as
\begin{align*}
	[Sf](x)=\int_{\Omega}k_{\theta}(x,t)f(t)\ud{t},
\end{align*}
for $f\in L^2(\Omega)$. Let $(\lambda_p,e_p)_{p\geq 1}$ be the eigenvalues and eigenfunctions of operator $S$. By Mercer's theorem \citep{FerMen:J09}, we have
\begin{align*}
	k_{\theta}(x,t)=\sum_{p\geq 1}\lambda_pe_p(x)e_p(t),\ \forall x,t\in\Omega.
\end{align*}

We examine the terms in the expansion of $\what{\mmd}^2[\wtil{\calK}_{\theta}(\Omega),\bbP,\bbQ]$, as given in \cref{eq:ker_dist_est}. Let $\what{e}_p(X)_m=\odfrac{m}\sum_{i=1}^me_p(X_i)$ and $\what{e}_p(Y)_n=\odfrac{n}\sum_{i=1}^ne_p(Y_i)$. 

For the first term in \cref{eq:ker_dist_est}, we have
\begin{align*}
	\ofrac{m^2}\sum_{i=1}^m\sum_{j=1}^mk_{\theta}(X_i,X_j)
	&=\ofrac{m^2}\sum_{i=1}^m\sum_{j=1}^m\sum_{p\geq 1}\lambda_pe_p(X_i)e_p(X_j)\nn
	&=\ofrac{m^2}\sum_{p\geq 1}\lambda_p\sum_{i=1}^m\sum_{j=1}^me_p(X_i)e_p(X_j)\nn
	&=\sum_{p\geq 1}\lambda_p\left(\what{e}_p(X)_m\right)^2.
\end{align*}
Similarly, for the second term in \cref{eq:ker_dist_est},
\begin{align*}
	&\ofrac{mn}\sum_{i=1}^m\sum_{j=1}^nk_{\theta}(X_i,Y_j)\nn
	&=\ofrac{mn}\sum_{i=1}^m\sum_{j=1}^n\sum_{p\geq 1}\lambda_pe_p(X_i)e_p(Y_j)\nn
	&=\ofrac{mn}\sum_{p\geq 1}\lambda_p\sum_{i=1}^m\sum_{j=1}^ne_p(X_i)e_p(Y_j)\nn
	&=\sum_{p\geq 1}\lambda_p\left(\sqrt{\dfrac{m}{n}}\what{e}_p(X)_m+\sqrt{\dfrac{n}{m}}\what{e}_p(Y)_n\right)^2-\left(\sqrt{\dfrac{m}{n}}\what{e}_p(X)_m\right)^2-\left(\sqrt{\dfrac{n}{m}}\what{e}_p(Y)_n\right)^2.
\end{align*}

By applying the law of large numbers, $\what{e}_p(X)_m$ converges almost surely to $\E[e_p(X)]$. This implies that $\what{\mmd}^2[\wtil{\calK}_{\theta}(\Omega),\bbP,\bbQ]\convas\mmd^2[\wtil{\calK}_{\theta}(\Omega),\bbP,\bbQ]$ by noting the expansion of $\mmd^2[\wtil{\calK}_{\theta}(\Omega),\bbP,\bbQ]$ in \cref{eq:ker_dist}. The proof is complete.

\subsection{Proof of \cref{appx:uniform_convas}}
Since $\Theta$ is compact, for any $\epsilon>0$, $\Theta$ has a finite cover $\{B(\theta_l,\epsilon)\}_{l=1}^L$ for some $L<\infty$, in which $B(\theta_l,\epsilon)$ is a ball of radius $\epsilon$ centered around $\theta_l$. Then we have
\begin{align*}
	&\abs*{\sup_{\theta\in\Theta}T_n(\theta)-\sup_{\theta\in\Theta}T(\theta)}\nn
	&\leq\sup_{\theta\in\Theta}\abs*{T_n(\theta)-T(\theta)}\nn
	&=\max_{1\leq l\leq L}\braces*{\sup_{\theta\in B(\theta_l,\epsilon)}\abs*{T_n(\theta)-T(\theta)}}\nn
	&=\max_{1\leq l\leq L}\braces*{\sup_{\theta\in B(\theta_l,\epsilon)}\abs*{T_n(\theta_l)-T(\theta_l)+T_n(\theta)-T_n(\theta_l)-T(\theta)+T(\theta_l)}}\nn
	&=\max_{1\leq l\leq L}\braces*{\abs*{T_n(\theta_i)-T(\theta_i)}+\sup_{\theta\in B(\theta_l,\epsilon)}\abs*{T_n(\theta)-T_n(\theta_l)}+\sup_{\theta\in B(\theta_l,\epsilon)}\abs*{T(\theta)-T(\theta_l)}}\nn
	&\leq\max_{1\leq l\leq L}\braces*{\abs*{T_n(\theta_l)-T(\theta_l)}+\sup_{\theta\in B(\theta_l,\epsilon)}\abs*{T_n(\theta)-T_n(\theta_l)}+\sup_{\theta\in B(\theta_l,\epsilon)}\abs*{T(\theta)-T(\theta_l)}}\nn
	&\leq\max_{1\leq l\leq L}\braces*{\abs*{T_n(\theta_l)-T(\theta_l)}+2K\epsilon}.
\end{align*}
The proof is complete due to the arbitrary smallness of $\epsilon$.

\section{Proof of \cref{thm:tv_stats} (b)}\label{appdx:tv_conv_rate}

Recall that $\calK_{\theta}$ denotes the RKHS associated with kernel $k_{\theta}$, where $\theta\in\Theta$. Recall $X$ and $X'$ are i.i.d. random variables both following the distribution $\bbP$, and $Y$ and $Y'$ are i.i.d. random variables following the distribution $\bbQ$. 

Utilizing the inequality $\abs*{\sup\{f(t):t\in A\}-\sup\{g(t):t\in A\}}\leq\sup\{\abs{f(t)-g(t):t\in A}\}$ and leveraging the expansions of $\what{\mmd}^2[\wtil{\calK}_{\theta}(\Omega),\bbP,\bbQ]$ and $\mmd^2[\wtil{\calK}_{\theta}(\Omega),\bbP,\bbQ]$ in \cref{eq:ker_dist,eq:ker_dist_est}, we obtain
\begin{align*}
	&\abs*{\sup_{\theta\in\Theta}\what{\mmd}^2[\wtil{\calK}_{\theta}(\Omega),\bbP,\bbQ]-\sup_{\theta\in\Theta}\mmd^2[\wtil{\calK}_{\theta}(\Omega),\bbP,\bbQ]} \nn
	&\leq\sup_{\theta\in\Theta}\abs*{T_m(X;\theta)}+\sup_{\theta\in\Theta}\abs*{T_n(Y;\theta)}-\sup_{\theta\in\Theta}\abs*{T_{m,n}(X,Y;\theta)},
\end{align*}
where
\begin{align*}
	&T_m(X;\theta)=\ofrac{M_{k_{\theta}}}\left(\odfrac{m^2}\sum_{i=1}^m\sum_{j=1}^mk_{\theta}(X_i,X_j)-\E_{X,X'}[k_{\theta}(X,X')]\right),\nn
	&T_n(X;\theta)=\ofrac{M_{k_{\theta}}}\left(\odfrac{n^2}\sum_{i=1}^n\sum_{j=1}^nk_{\theta}(Y_i,Y_j)-\E_{X,X'}[k_{\theta}(Y,Y')]\right),\nn
	&T_{m,n}(X,Y;\theta)=\ofrac{M_{k_{\theta}}}\left(\dfrac{2}{mn}\sum_{i=1}^m\sum_{j=1}^nk_{\theta}(X_i,Y_j)-2\E_{X,Y}[k_{\theta}(X,Y)]\right).
\end{align*}

In what follows, we establish the convergence rates for $\sup_{\theta\in\Theta}\abs*{T_m(X;\theta)}$, $\sup_{\theta\in\Theta}\abs*{T_n(Y;\theta)}$ and $\sup_{\theta\in\Theta}\abs*{T_{m,n}(X,Y;\theta)}$ separately. Then, the final result can be obtained using the following inequality:
\begin{align*}
	\P(A+B\geq t+s)\leq\P(A\geq t)+\P(B\geq s),
\end{align*}
which holds due to $\{A+B\geq t+s\}\subseteq\{\{A<t\}\cap\{B<s\}\}^\complement=\{A\geq t\}\cup\{A\geq t\}$.

\paragraph{Step 1}
In this step, we derive the convergence rate for $T_m(X;\theta)$, $T_n(Y;\theta)$ and $T_{m,n}(X,Y;\theta)$ by employing the McDiarmid's inequality \citep{Mic:J89}.

\begin{Theorem}[McDiarmid's inequality]
	Suppose $f:\bbR^d\mapsto\bbR$ satisfies the bounded difference condition: there exists constants $c_1,\ldots,c_n\in\bbR$ such that for all real numbers $x_1,\ldots,x_n$ and $x_i'$,
	\begin{align*}
		\abs{f(x_1,\ldots,x_n)-f(x_1,\ldots,x_{i-1},x_i',x_{i+1},\ldots,x_n)}\leq c_i.
	\end{align*}
	Then for any i.i.d. random variables $X_1,\ldots,X_n$,
	\begin{align*}
		\P\left(\abs{f(X_1,\ldots,X_n)-\E[f(X_1,\ldots,X_n)]}\geq t\right)\leq\exp\left(-\dfrac{2t^2}{\sum_{i=1}^nc_i^2}\right).
	\end{align*}
\end{Theorem}

Denoting $\wtil{T}_m(X;\theta)$ as
\begin{align*}
	\wtil{T}_m(X;\theta)=\ofrac{M_{k_{\theta}}}\left(\odfrac{m(m-1)}\sum_{i=1}^m\sum_{j\neq i}^mk_{\theta}(X_i,X_j)-\E_{X,X'}[k_{\theta}(X,X')]\right),
\end{align*}
we proceed to apply the McDiarmid's inequality to $\wtil{T}_m(X;\theta)$.

Due to $\odfrac{M_{k_{\theta}}}k_{\theta}(X_i,X_j)\leq 1$, $c_i\leq\dfrac{2m-2}{m(m-1)}=\dfrac{2}{m}$ for all $1\leq i\leq m$. Therefore, we have
\begin{align*}
	\P(\abs*{\wtil{T}_m(X;\theta)}\geq t)\leq\exp\left(-\dfrac{m}{2}t^2\right).
\end{align*}

Moreover, we have
\begin{align*}
	\abs*{T_m(X;\theta)}
	&=\abs*{\dfrac{m-1}{m}\wtil{T}_m(X;\theta)-\odfrac{m}\E_{X,X'}[\ofrac{M_{k_{\theta}}}k(X,X')]+\odfrac{m^2}\sum_{i=1}^m\ofrac{M_{k_{\theta}}}k(X_i,X_i)}\nn
	&\leq\dfrac{m-1}{m}\abs*{\wtil{T}_m(X;\theta)}+\dfrac{2}{m}.
\end{align*}
Therefore, we have
\begin{align*}
	\P(\abs*{T_m(X;\theta)}\geq\dfrac{m-1}{m}t+\dfrac{2}{m})\leq\P(\abs*{\wtil{T}_m(X;\theta)}\geq t)\leq\exp\left(-\dfrac{m}{2}t^2\right).
\end{align*}

Using the same reasoning, we have
\begin{align*}
	&\P(\abs*{T_n(Y;\theta)}\geq\dfrac{n-1}{n}t+\dfrac{2}{n})\leq\exp\left(-\dfrac{n}{2}t^2\right).
\end{align*}

For $T_{m,n}(X,Y;\theta)$, we can directly apply the McDiarmid's inequality, where $c_i=\odfrac{m}$ for $1\leq i\leq m$ and $c_i=\odfrac{n}$ for $(n+1)\leq i\leq (m+n)$. Hence $\sum_{i=1}^{m+n}c_i^2\leq\odfrac{m}+\odfrac{n}$. We have
\begin{align*}
	\P(\abs*{T_{m,n}(X,Y;\theta)}\geq t)\leq\exp\left(-\dfrac{mn}{m+n}t^2\right).
\end{align*}

\paragraph{Step 2}
At this stage, we show that $\sup_{\theta\in\Theta}\abs*{T_m(X;\theta)}$ attains the same convergence rate as $T_m(X;\theta)$ when $\Theta$ is compact. 

Firstly, we prove that for any $\delta>0$, there exists $L<\infty$ such that
\begin{align*}
	\sup_{\theta\in\Theta}\abs*{T_m(X;\theta)}\leq\max_{1\leq l\leq L}\abs*{T_m(X;\theta_l)}+\delta,
\end{align*}
where
\begin{align*}
	T_m(X;\theta_l)
	=\odfrac{m^2}\sum_{i=1}^m\sum_{j=1}^m\sup_{\theta\in B_{\epsilon}(\theta_l)}\dfrac{k_{\theta}(X_i,X_j)}{M_{k_{\theta}}}-\E_{X,X'}[\sup_{\theta\in B_{\epsilon}(\theta_l)}\dfrac{k_{\theta}(X,X')}{M_{k_{\theta}}}],
\end{align*}
with $B_{\epsilon}(\theta_l)=\braces*{\theta\in\Theta:\norm*{\theta-\theta_l}_2\leq\epsilon}$.

\begin{proof}
Given the compactness of $\Theta$, for any $\epsilon>0$, we can identify a finite cover $\bigcup_{l=1}^L B_{\epsilon}(\theta_l)$ such that $\Theta\subseteq\bigcup_{l=1}^L B_{\epsilon}(\theta_l)$, where $\theta_l\in\Theta$. Adapting the proof technique from the Glivenko-Cantelli theorem \citep{Dud:B02}, we obtain
\begin{align*}
	&\sup_{\theta\in\Theta}\abs*{T_m(X;\theta)}\nn
	&=\max_{1\leq l\leq L}\sup_{\theta\in B_{\epsilon}(\theta_l)}\abs*{T_m(X;\theta)}\nn
	&\leq\max_{1\leq l\leq L}\abs*{\odfrac{m^2}\sum_{i=1}^n\sum_{j=1}^n\sup_{\theta\in B_{\epsilon}(\theta_l)}\dfrac{k_{\theta}(X_i,X_j)}{M_{k_{\theta}}}-\E_{X,X'}[\inf_{\theta\in B_{\epsilon}(\theta_l)}\dfrac{k_{\theta}(X,X')}{M_{k_{\theta}}}]}\nn
	&\leq\max_{1\leq l\leq L}\abs*{T_m(X;\theta_l)}+\max_{1\leq l\leq L}\abs*{\E_{X,X'}[\triangle_l(X,X')]},
\end{align*}
where
\begin{align*}
	\triangle_l(X,X')
	&=\sup_{\theta\in B_{\epsilon}(\theta_l)}\dfrac{k_{\theta}(X,X')}{M_{k_{\theta}}}-\inf_{\theta\in B_{\epsilon}(\theta_l)}\dfrac{k_{\theta}(X,X')}{M_{k_{\theta}}}\nn
	&\leq 2\sup_{\theta\in B_{\epsilon}(\theta_l)}\abs*{\dfrac{k_{\theta}(X,X')}{M_{k_{\theta}}}-\dfrac{k_{\theta_l}(X,X')}{M_{k_{\theta_l}}}}\nn
	&\leq 2\sup_{\theta\in B_{\epsilon}(\theta_l)}\braces*{C\norm*{\theta-\theta_l}_2}\leq 2C\epsilon.
\end{align*}
The proof is complete by noting $\epsilon$ can be chosen to be sufficiently small.
\end{proof}

Furthermore, we have
\begin{align*}
	\P(\sup_{\theta\in\Theta}\abs*{T_m(X;\theta)}\geq t)
	&\leq\P(\max_{1\leq l\leq L}\abs*{T_m(X;\theta_l)}\geq t)\nn
	&\leq\P(\abs*{T_m(X;\theta_l)}\geq t),
\end{align*}
for any $1\leq l\leq L$. The proof is complete by noting that $T_m(X;\theta_l)$ achieves the same convergence rate as $T_m(X;\theta)$.

\section{Proof of \cref{thm:theta_stats} (a)}\label{appx:theta_convd}
The proof relies on the asymptotic distribution of U-statistics \citep{Hoe:J89}. For a random variable $Z$, consider a sequence of i.i.d. samples $(Z_i)_{i=1}^n$. Let $Z'$ be a random variable following the same distribution as $Z$. Given a symmetric function $h(Z,Z')$, an unbiased estimator of $\E_{Z,Z'}[h(Z,Z')]$ is
\begin{align*}
	\what{U}_n(h)=\odfrac{n(n-1)}\sum_{i=1}^n\sum_{j\neq i}^nh(Z_i,Z_j).
\end{align*}

\begin{Theorem}\label{thm:asymp_u_stats}
	Let $\sigma^2=\var\left(\E_{Z}[h(Z,Z')]\right)$. If $\sigma^2<\infty$, then
	\begin{align*}
		\sqrt{n}\left(\what{U}_n(X)-\E_{Z,Z'}[h(Z,Z')]\right)\convd\N{0}{4\sigma^2}.
	\end{align*}
\end{Theorem}

We return to the main proof. Let $Z_i=(X_i,Y_i)$ and hence $Z_1,\ldots,Z_n$ are independently and identically distributed. Denote
\begin{align*}
	h_{\theta}(Z_i,Z_j)=\ofrac{M_{k_{\theta}}}\left(k_{\theta}(X_i,X_j)-2k_{\theta}(X_i,Y_j)+k_{\theta}(Y_i,Y_j)\right).
\end{align*}

When $m=n$, $\what{L}(\theta)$ can be written as
\begin{align}
	\what{L}(\theta)
	&=\what{\mmd}^2[\wtil{\calK}_{\theta}(\Omega),\bbP,\bbQ]\nn
	&=\odfrac{n^2}\sum_{i=1}^n\sum_{j=1}^nh_{\theta}(Z_i,Z_j)\nn \label{eq:L_expand}
	&=\dfrac{n-1}{n}\left(\odfrac{n(n-1)}\sum_{i=1}^n\sum_{j\neq i}^nh_{\theta}(Z_i,Z_j)\right)+\odfrac{n^2}\sum_{i=1}^nh_{\theta}(Z_i,Z_i).
\end{align}

Recall that $\nabla\what{L}(\what{\theta})$ and $\nabla^2\what{L}(\what{\theta})$ denote the gradient and Hessian of $\what{L}(\theta)$ w.r.t. $\theta$, evaluated at $\what{\theta}$, receptively. Note that $\nabla\what{L}(\what{\theta})=0$ by assumption. The second-order Taylor expansion of $\nabla\what{L}(\what{\theta})$ with the Lagrange remainder is given by
\begin{align}\label{eq:taylor_L_hat}
	\nabla\what{L}(\what{\theta})=\nabla\what{L}(\theta^*)+\nabla^2\what{L}(\theta^*)(\what{\theta}-\theta^*)+(\what{\theta}-\theta^*)\T M(\what{\theta}-\theta^*).
\end{align}
where $M$ is a matrix where the $(i,j)$-th entry is
\begin{align*}
	M_{i,j}=\left(\nabla\left(\nabla^2\what{L}(\theta)\right)_{i,j}(\theta')\right)(\what{\theta}-\theta^*),
\end{align*}
for some fixed $\theta'$. Rearranging the equation \cref{eq:taylor_L_hat}, we obtain
\begin{align*}
	\sqrt{n}(\what{\theta}-\theta^*)\left(\bI-\calO(\norm*{\what{\theta}-\theta^*})\right)
	=-\inv{\nabla^2\what{L}(\theta^*)}\left(\sqrt{n}\nabla\what{L}(\theta^*)\right).
\end{align*}

From \cref{eq:L_expand}, $\nabla\what{L}(\theta^*)$ can be written as
\begin{align}\label{eq:grad_L_expand}
	\nabla\what{L}(\theta^*)=\dfrac{n-1}{n}\underbrace{\left(\odfrac{n(n-1)}\sum_{i=1}^n\sum_{j\neq i}^n\nabla h_{\theta^*}(Z_i,Z_j)\right)}_{U(\nabla h_{\theta^*})}+\odfrac{n^2}\sum_{i=1}^n\nabla h_{\theta^*}(Z_i,Z_i).
\end{align}

According to \cref{thm:asymp_u_stats}, we obtain the following asymptotic result:
\begin{align*}
	\sqrt{n}\left(U(\nabla h_{\theta^*})-\E_{Z,Z'}[h_{\theta^*}(Z,Z')]\right)
	\convd\N{0}{4\bSigma},
\end{align*}
where $\bSigma=\cov\left(\E_{Z'}[\nabla h_{\theta^*}(Z,Z')]\right)$.
 
Taking into account that $\E_{Z,Z'}[h_{\theta^*}(Z,Z')]=L(\theta^*)=0$ and the last term in \cref{eq:grad_L_expand} converges to $0$, we deduce that  $\sqrt{n}\nabla\what{L}(\theta^*)\convd\N{0}{4\bSigma}$.

Applying the techniques outlined in \cref{appx:tv_convas}, it can be demonstrated that $\nabla^2\what{L}(\theta^*)\convas\nabla^2L(\theta^*)$. Consequently, we arrive at the final result:
\begin{align*}
	\sqrt{n}(\what{\theta}-\theta^*)
	\convd\N{0}{4\left(\nabla^2L(\theta^*)\right)^{-1}\bSigma\left(\nabla^2L(\theta^*)\right)^{-1}}.
\end{align*}

\section{Proof of \cref{thm:theta_stats} (b)}\label{appx:theta_conv_rate}
Recall the notation $\ell(\theta;x,y)=\dfrac{k_{\theta}(x,y)}{M_{k_{\theta}}}$. Under the condition $\norm*{\nabla^2\ell(\theta;x,y)}\geq\gamma$ for all $\theta\in\Theta$ and $x,y\in\Omega$, it is straightway to verify that $\norm*{\nabla^2\what{L}(\theta)}\geq\gamma$.

The first-order Taylor expansion of $\nabla\what{L}(\what{\theta})$ with the Lagrange remainder is given by
\begin{align*}
	\nabla\what{L}(\what{\theta})=\nabla\what{L}(\theta^*)+\nabla^2\what{L}(\theta')(\what{\theta}-\theta^*),
\end{align*}
for some $\theta'\in\Theta$. Consequently we have
\begin{align*}
	(\what{\theta}-\theta^*)
	=\left(\nabla^2\what{L}(\theta')\right)^{-1}\left(\nabla\what{L}(\what{\theta})-\nabla\what{L}(\theta^*)\right).
\end{align*}

Subsequently, we obtain
\begin{align*}
	\norm*{\what{\theta}-\theta^*}\leq\norm*{\nabla^2\what{L}(\theta')}^{-1}\norm*{\nabla\what{L}(\theta^*)}
	\leq\gamma^{-1}\norm*{\nabla\what{L}(\theta^*)}.
\end{align*}

From \cref{eq:ker_dist_est}, we obtain
\begin{align*}
	\ppfrac{\what{L}(\theta^*)}{\theta_i}
	&=\frac{1}{m^2}\sum_{i=1}^m\sum_{j=1}^m\ppfrac{\ell(\theta;X_i,X_j)}{\theta_i}
	+\frac{1}{n^2}\sum_{i=1}^n\sum_{i=1}^n\ppfrac{\ell(\theta;Y_i,Y_j)}{\theta_i}\nn
	&-\frac{2}{mn}\sum_{i=1}^m\sum_{i=1}^n\ppfrac{\ell(\theta;X_i,Y_j)}{\theta_i}.
\end{align*}

Noting $\abs*{\ppfrac{\ell(\theta;x,y)}{\theta_i}}\leq\tau_i$ and applying the techniques in \cref{appdx:tv_conv_rate}, we can easily obtain
\begin{align*}
	&\P(\abs*{\ppfrac{\what{L}(\theta^*)}{\theta_i}}\geq 3t+\dfrac{2\tau_i}{m}+\dfrac{2\tau_i}{n})\nn
	&\leq\exp\left(-\dfrac{m}{2}\dfrac{t^2}{\tau_i^2}\right)+\exp\left(-\dfrac{n}{2}\dfrac{t^2}{\tau_i^2}\right)+\exp\left(-\dfrac{mn}{m+n}\dfrac{t^2}{\tau_i^2}\right),
\end{align*}
for any $t>0$.

Finally, noting $\norm*{\nabla_{\theta}\what{L}(\theta^*)}=\sqrt{\sum_{i=1}^s\abs*{\ppfrac{\what{L}(\theta^*)}{\theta_i}}^2}$, we thus have
\begin{align*}
	&\P(\norm*{\what{\theta}-\theta^*}_2\geq\gamma^{-1}\sqrt{\sum_{i=1}^s\left(3t+\dfrac{2\tau_i}{m}+\dfrac{2\tau_i}{n}\right)^2})\nn
	&\leq\sum_{i=1}^s\left(\exp\left(-\dfrac{mt^2}{2\tau_i^2}\right)+\exp\left(-\dfrac{nt^2}{2\tau_i^2}\right)+\exp\left(-\dfrac{mnt^2}{(m+n)\tau_i^2}\right)\right).
\end{align*}
The proof is complete.

\section{Proof of \cref{thm:chi2_conv} (a)}\label[Appendix]{appx:chi2_convas}
To commence, we demonstrate that the eigenvalues of $\Psi_{\lambda}$ and $\what{\Psi}_{\lambda}$ lie in the range of $\left[\lambda,\lambda+M_k\right]$. To see this, note that both $\Psi$ and $\what{\Psi}$ are Hilbert-Schmidt operators. For instance, we have
\begin{align*}
	\Tr(\Psi)=\sum_{i\geq 1}\ip*{e_i}{\Psi e_i}_{\calK}
	=\E_{\bbQ}[\ip*{\Phi(X)}{e_i}_{\calK}^2]
	=\E_{\bbQ}[\norm*{\Phi(X)}_{\calK}^2]=\E_{\bbQ}[k(X,X)]<\infty.
\end{align*}
Therefore, the eigenvalues, in non-increasing order, tend to be $0$. On the other hand, for any $f\in\calK$, we have
\begin{align*}
	\dfrac{\ip*{f}{\Psi f}_{\calK}}{\norm*{f}_{\calK}}
	=\dfrac{\E_{Y}[f(Y)^2]}{\norm*{f}_{\calK}^2}
	&=\dfrac{\E_{Y}[\ip*{f}{\Phi(Y)}_{\calK}^2]}{\norm*{f}_{\calK}^2}\nn
	&\leq\dfrac{\norm*{f}_{\calK}^2\E_{Y}[k(Y,Y)]}{\norm*{f}_{\calK}^2}
	\leq M_k.
\end{align*}
Similarly, the same upper-bound reasoning applies to $\what{\Psi}$. As a result, we obtain
\begin{align*}
	&\opnorm{\Psi_{\lambda}}\leq\lambda+M_k,\opnorm{\Psi_{\lambda}^{-1}}\leq\lambda^{-1},\ \text{and} \nn
	&\opnorm{\what{\Psi}_{\lambda}}\leq\lambda+M_k,\opnorm{\what{\Psi}_{\lambda}^{-1}}\leq\lambda^{-1}.
\end{align*}

Denote $\mu=\mu_{\bbP}-\mu_{\bbQ}$ and $\what{\mu}=\what{\mu}_{\bbP}-\what{\mu}_{\bbQ}$. Moreover, we have $\norm*{\mu}_{\calK}\leq 2\sqrt{M_k}$ and $\norm*{\what{\mu}}_{\calK}\leq 2\sqrt{M_k}$, both bounded in $\calK$ under \cref{asm:int_ker}. This is shown as follows:
\begin{align*}
	\norm*{\mu}_{\calK}
	&=\norm*{\E_{X}[\Phi(X)]-\E_{Y}[\Phi(Y)]}_{\calK}\nn
	&\leq\norm*{\E_{X}[\Phi(X)]}_{\calK}+\norm*{\E_{Y}[\Phi(Y)]}_{\calK}\nn
	&=\sqrt{\E_{X,X'}[k(X,X')]}+\sqrt{\E_{Y,Y'}[k(Y,Y')]}\nn
	&\leq 2\sqrt{M_k},
\end{align*}
and
\begin{align*}
	\norm*{\what{\mu}}_{\calK}
	&=\norm*{\odfrac{m}\sum_{i=1}^m\Phi(X_i)-\odfrac{n}\sum_{i=1}^n\Phi(Y_i)}_{\calK}\nn
	&\leq\norm*{\odfrac{m}\sum_{i=1}^m\Phi(X_i)}_{\calK}+\norm*{\odfrac{n}\sum_{i=1}^n\Phi(Y_i)}_{\calK}\nn
	&=\odfrac{m}\sqrt{\sum_{i=1}^m\sum_{j=1}^mk(X_i,X_j)}+\odfrac{n}\sqrt{\sum_{i=1}^n\sum_{j=1}^nk(Y_i,Y_j)}\nn
	&\leq 2\sqrt{M_k}.
\end{align*}

Importantly, note that for $\lambda>0$, both $\what{\Psi}_{\lambda}$ and $\Psi_{\lambda}$, along with their inverses, are bounded self-adjoint operators.

Now we return to the main proof. 

Applying the triangle inequality, we obtain
\begin{align}
	\abs*{\ip*{\mu}{\Psi_{\lambda}^{-1}\mu}_{\calK}-\ip*{\what{\mu}}{\what{\Psi}_{\lambda}^{-1}\what{\mu}}_{\calK}} \label{chi2_as_sub1}
	&\leq\abs*{\ip*{\mu}{\Psi_{\lambda}^{-1}\mu}_{\calK}-\ip*{\mu}{\what{\Psi}_{\lambda}^{-1}\mu}_{\calK}}\\ \label{chi2_as_sub2}
	&+\abs*{\ip*{\mu}{\what{\Psi}_{\lambda}^{-1}\mu}_{\calK}-\ip*{\mu}{\what{\Psi}_{\lambda}^{-1}\what{\mu}}_{\calK}}\\ 
	&+\abs*{\ip*{\mu}{\what{\Psi}_{\lambda}^{-1}\what{\mu}}_{\calK}-\ip*{\what{\mu}}{\what{\Psi}_{\lambda}^{-1}\what{\mu}}_{\calK}}. \label{chi2_as_sub3}
\end{align}

It suffices to prove the almost surely (a.s.) convergence for each of \cref{chi2_as_sub1,chi2_as_sub2,chi2_as_sub3}.

For \cref{chi2_as_sub2}, the a.s. convergence can be established using the following inequality:
\begin{align*}
	\abs*{\ip*{\mu}{\what{\Psi}_{\lambda}^{-1}\mu}_{\calK}-\ip*{\mu}{\what{\Psi}_{\lambda}^{-1}\what{\mu}}_{\calK}}
	&=\abs*{\ip*{\what{\Psi}_{\lambda}^{-1}\mu}{\mu}_{\calK}-\ip*{\what{\Psi}_{\lambda}^{-1}\mu}{\what{\mu}}_{\calK}}\nn
	&=\abs*{\ip*{\what{\Psi}_{\lambda}^{-1}\mu}{\mu-\what{\mu}}_{\calK}}\nn
	&\leq\norm*{\what{\Psi}_{\lambda}^{-1}\mu}_{\calK}\norm*{\mu-\what{\mu}}_{\calK},
\end{align*}
where $\norm*{\what{\Psi}_{\lambda}^{-1}\mu}_{\calK}\leq\opnorm{\what{\Psi}_{\lambda}^{-1}}\norm*{\mu}_{\calK}$ is bounded. Then, the convergence of \cref{chi2_as_sub2} follows from $\norm*{\mu-\what{\mu}}_{\calK}^2\convas 0$, which is shown in \cref{appdx:convas_mu_psi}.  The proof of a.s. convergence for \cref{chi2_as_sub3} follows a similar technique.

Now it remains to be proved that that $\cref{chi2_as_sub1}\convas 0$.

Using the result that if two operators $A$ and $B$ are invertible, then $A^{-1}-B^{-1}=B^{-1}(B-A)A^{-1}$, we obtain
\begin{align*}
	\abs*{\ip*{\mu}{\what{\Psi}_{\lambda}^{-1}\mu}_{\calK}-\ip*{\mu}{\Psi_{\lambda}^{-1}\mu}_{\calK}}
	&=\abs*{\ip*{\mu}{(\what{\Psi}_{\lambda}^{-1}-\Psi_{\lambda}^{-1})\mu}_{\calK}}\nn
	&=\abs*{\ip*{\mu}{\Psi_{\lambda}^{-1}\left(\Psi_{\lambda}-\what{\Psi}_{\lambda}\right)\what{\Psi}_{\lambda}^{-1}\mu}_{\calK}}\nn
	&=\abs*{\ip*{\left(\Psi_{\lambda}-\what{\Psi}_{\lambda}\right)\Psi_{\lambda}^{-1}\mu}{\what{\Psi}_{\lambda}^{-1}\mu}_{\calK}}\nn
	&\leq\norm*{\left(\Psi_{\lambda}-\what{\Psi}_{\lambda}\right)\Psi_{\lambda}^{-1}\mu}_{\calK}\norm*{\what{\Psi}_{\lambda}^{-1}\mu}_{\calK}\nn
	&\leq\norm*{\left(\Psi_{\lambda}-\what{\Psi}_{\lambda}\right)h}_{\calK}\opnorm{\what{\Psi}_{\lambda}^{-1}}\norm*{\mu}_{\calK}\nn
	&\leq\lambda^{-1}M_k\norm*{\left(\Psi_{\lambda}-\what{\Psi}_{\lambda}\right)h}_{\calK},
\end{align*}
where $h=\Psi_{\lambda}^{-1}\mu$ is a deterministic element in $\calK$. The proof of $\norm*{\left(\Psi_{\lambda}-\what{\Psi}_{\lambda}\right)h}_{\calK}\convas 0$ is provided in \cref{appdx:convas_mu_psi}.
The proof is complete.

\subsection{Almost surely convergence of $\norm*{\mu-\what{\mu}}_{\calK}$ and $\norm*{(\Psi_{\lambda}-\what{\Psi}_{\lambda})h}_{\calK}$}
\label{appdx:convas_mu_psi}

We have the following bound on the norm of the difference between $\mu$ and $\what{\mu}$ in $\calK$:
\begin{align*}
	\norm*{\mu-\what{\mu}}_{\calK}^2
	\leq\norm*{\what{\mu}_{\bbP}-\mu_{\bbP}}_{\calK}^2+\norm*{\what{\mu}_{\bbQ}-\mu_{\bbQ}}_{\calK}^2,
\end{align*}
where the terms can be expanded as
\begin{align}
	&\norm*{\what{\mu}_{\bbP}-\mu_{\bbP}}_{\calK}^2\nn \label{mu_u_stats}
	&=\dfrac{1}{m^2}\sum_{i=1}^m\sum_{j=1}^mk(X_i,X_j)-\E_{X,X'}k(X,X')\\ 
	&-\dfrac{2}{m}\sum_{i=1}^m\E_{X}[k(X,X_i)]+2\E_{X,X'}k(X,X').\label{mu_mean_stats}
\end{align}

We also have
\begin{align}
	&\norm*{(\what{\Psi}_{\lambda}-\Psi_{\lambda})h}_{\calK}^2\nn \label{psi_u_stats}
	&=\odfrac{n^2}\sum_{i=1}^n\sum_{j=1}^nk(Y_i,Y_j)h(Y_i)h(Y_j)-\E_{Y,Y'}[k(Y,Y')h(Y)h(Y')] \\
	&-\dfrac{2}{n}\sum_{i=1}^n\E_{Y}[k(Y_i,Y)h(Y_i)h(Y)]+2\E_{Y,Y'}[k(Y,Y')h(Y)h(Y')]. \label{psi_mean_stats}
\end{align}

The almost surely convergence of \cref{mu_u_stats,psi_u_stats} can be referred to \cref{appx:tv_convas}, where the eigendecomposition of the kernel is applied. Additionally, the almost surely convergence of \cref{mu_mean_stats,psi_mean_stats} follow from the law of large numbers.

\section{Proof of \cref{thm:chi2_conv} (b)}\label{appdx:chi2_convd}
Let $\mu=\mu_{\bbP}-\mu_{\bbQ}$ and $\what{\mu}=\what{\mu}_{\bbP}-\what{\mu}_{\bbQ}$. Recall $\opnorm{\Psi_{\lambda}^{-1}}\leq\lambda^{-1}$ and $\opnorm{\what{\Psi}_{\lambda}^{-1}}\leq\lambda^{-1}$ from \cref{appx:chi2_convas}.

In step 1, we show the convergence of distribution of the random variable $\ip*{\what{\mu}}{\what{\Psi}_{\lambda}^{-1}\what{\mu}}_{\calK}$ to $\ip*{\what{\mu}}{\Psi_{\lambda}^{-1}\what{\mu}}_{\calK}$. In step 2, we derive the asymptotic distribution for $\ip*{\what{\mu}}{\Psi_{\lambda}^{-1}\what{\mu}}_{\calK}$.

\paragraph{Step 1} 
From the triangle inequality, we decompose the differences as follows:
\begin{align}\label{chi2_convd_sub1}
	\abs*{\ip*{\what{\mu}}{\what{\Psi}_{\lambda}^{-1}\what{\mu}}_{\calK}-\ip*{\what{\mu}}{\Psi_{\lambda}^{-1}\what{\mu}}_{\calK}}
	&\leq\abs*{\ip*{\what{\mu}}{\what{\Psi}_{\lambda}^{-1}\what{\mu}}_{\calK}-\ip*{\mu}{\what{\Psi}_{\lambda}^{-1}\what{\mu}}_{\calK}}\\ \label{chi2_convd_sub2}
	&+\abs*{\ip*{\mu}{\Psi_{\lambda}^{-1}\what{\mu}}_{\calK}-\ip*{\what{\mu}}{\Psi_{\lambda}^{-1}\what{\mu}}_{\calK}}\\ \label{chi2_convd_sub3}
	&+\abs*{\ip*{\mu}{\what{\Psi}_{\lambda}^{-1}\what{\mu}}_{\calK}-\ip*{\mu}{\Psi_{\lambda}^{-1}\what{\mu}}_{\calK}}.
\end{align}

For \cref{chi2_convd_sub1,chi2_convd_sub2}, we have
\begin{align*}
	\cref{chi2_convd_sub1}=\abs*{\ip*{\what{\mu}-\mu}{\what{\Psi}_{\lambda}^{-1}\what{\mu}}_{\calK}}
	&\leq\norm*{\what{\mu}-\mu}_{\calK}\norm*{\what{\Psi}_{\lambda}^{-1}\what{\mu}}_{\calK}, \nn
	&\leq\norm*{\what{\mu}-\mu}_{\calK}\opnorm{\what{\Psi}_{\lambda}^{-1}}\norm*{\what{\mu}}_{\calK} \nn
	&\leq\lambda^{-1}M_k\norm*{\what{\mu}-\mu}_{\calK}
\end{align*}
and
\begin{align*}
	\cref{chi2_convd_sub2}=\abs*{\ip*{\mu-\what{\mu}}{\Psi_{\lambda}^{-1}\what{\mu}}_{\calK}}
	&\leq\norm*{\what{\mu}-\mu}_{\calK}\norm*{\Psi_{\lambda}^{-1}\what{\mu}}_{\calK} \nn
	&\leq\norm*{\what{\mu}-\mu}_{\calK}\opnorm{\Psi_{\lambda}^{-1}}\norm*{\what{\mu}}_{\calK} \nn
	&\leq\lambda^{-1}M_k\norm*{\what{\mu}-\mu}_{\calK},
\end{align*}
where $\norm*{\what{\mu}-\mu}_{\calK}$ converges to $0$ as shown in \cref{appdx:convas_mu_psi}.

For \cref{chi2_convd_sub3}, we have
\begin{align*}
	\abs*{\ip*{\mu}{\what{\Psi}_{\lambda}^{-1}\what{\mu}}_{\calK}-\ip*{\mu}{\Psi_{\lambda}^{-1}\what{\mu}}_{\calK}}
	&=\abs*{\ip*{\mu}{\left(\what{\Psi}_{\lambda}^{-1}-\Psi_{\lambda}^{-1}\right)\what{\mu}}_{\calK}}\nn
	&=\abs*{\ip*{\mu}{\Psi_{\lambda}^{-1}\left(\Psi_{\lambda}-\what{\Psi}_{\lambda}\right)\what{\Psi}_{\lambda}^{-1}\what{\mu}}_{\calK}}\nn
	&=\abs*{\ip*{\left(\Psi_{\lambda}-\what{\Psi}_{\lambda}\right)\Psi_{\lambda}^{-1}\mu}{\what{\Psi}_{\lambda}^{-1}\what{\mu}}_{\calK}}\nn
	&\leq\norm*{\left(\Psi_{\lambda}-\what{\Psi}_{\lambda}\right)\Psi_{\lambda}^{-1}\mu}_{\calK}\norm*{\what{\Psi}_{\lambda}^{-1}\what{\mu}}_{\calK}\nn
	&\leq\norm*{\left(\Psi_{\lambda}-\what{\Psi}_{\lambda}\right)\Psi_{\lambda}^{-1}\mu}_{\calK}\opnorm{\what{\Psi}_{\lambda}^{-1}}\norm*{\what{\mu}}_{\calK}\nn
	&\leq\lambda^{-1}M_k\norm*{\left(\Psi_{\lambda}-\what{\Psi}_{\lambda}\right)\Psi_{\lambda}^{-1}\mu}_{\calK},
\end{align*}
where $\norm*{\left(\what{\Psi}_{\lambda}-\Psi_{\lambda}\right)\Psi_{\lambda}^{-1}\mu}_{\calK}\convas 0$ is shown in \cref{appdx:convas_mu_psi}. The proof of step 1 is complete.

\paragraph{Step 2}
Now we establish the asymptotic distribution for $\ip*{\what{\mu}}{\Psi_{\lambda}^{-1}\what{\mu}}_{\calK}$. We have
\begin{align*}
	\ip*{\what{\mu}}{\Psi_{\lambda}^{-1}\what{\mu}}_{\calK}
	&=\ip*{\sum_{i=1}^n\left(\Phi(X_i)-\Phi(Y_i)\right)}{\Psi_{\lambda}^{-1}\sum_{j=1}^n\left(\Phi(X_j)-\Phi(Y_j)\right)}_{\calK}\nn
	&=\sum_{i=1}^n\sum_{j=1}^n\ip*{\Phi(X_i)-\Phi(Y_i)}{\Psi_{\lambda}^{-1}\left(\Phi(X_j)-\Phi(Y_j)\right)}_{\calK}
\end{align*}

Applying the techniques presented in \cref{appx:theta_convd} (which is based on the asymptotic distribution of U-statistics), we obtain
\begin{align*}
	\sqrt{n}\left(\ip*{\what{\mu}}{\Psi_{\lambda}^{-1}\what{\mu}}_{\calK}-\ip*{\mu}{\Psi_{\lambda}^{-1}\mu}_{\calK}\right)\convd\N{0}{\sigma^2},
\end{align*}
where
\begin{align*}
	\sigma^2
	&=\var\left(\E_{X',Y'}[\ip*{\Phi(X)-\Phi(Y)}{\Psi_{\lambda}^{-1}\left(\Phi(X')-\Phi(Y')\right)}_{\calK}]\right)\nn
	&=\var\left(\ip*{\Phi(X)-\Phi(Y)}{\Psi_{\lambda}^{-1}\left(\mu_{\bbP}-\mu_{\bbQ}\right)}_{\calK}\right).
\end{align*}

\section{Proof of \cref{thm:chi2_conv} (c)}\label{appx:chi2_conv_rate}
Using the results in \cref{appx:chi2_convas}, we obtain
\begin{align}
	\abs*{\ip*{\what{\mu}}{\what{\Psi}_{\lambda}^{-1}\what{\mu}}_{\calK}-\ip*{\mu}{\Psi_{\lambda}^{-1}\mu}_{\calK}} \label{chi2_rate_a}
	&\leq\abs*{\ip*{\what{\mu}}{\what{\Psi}_{\lambda}^{-1}\what{\mu}}_{\calK}-\ip*{\mu}{\what{\Psi}_{\lambda}^{-1}\what{\mu}}_{\calK}} \\ \label{chi2_rate_b}
	&+\abs*{\ip*{\mu}{\what{\Psi}_{\lambda}^{-1}\what{\mu}}_{\calK}-\ip*{\mu}{\Psi_{\lambda}^{-1}\what{\mu}}_{\calK}}\\ \label{chi2_rate_c}
	&+\abs*{\ip*{\mu}{\Psi_{\lambda}^{-1}\what{\mu}}_{\calK}-\ip*{\mu}{\Psi_{\lambda}^{-1}\mu}_{\calK}}.
\end{align}

In what follows, we establish the convergence rate for each of \cref{chi2_rate_a,chi2_rate_b,chi2_rate_c}. The final result can be obtained using the inequality:
\begin{align*}
	\P(A+B\geq a+b)\leq\P(A\geq a)+\P(B\geq b).
\end{align*}

Firstly, recall $\lambda\leq\opnorm{\Psi_{\lambda}}\leq M_k+\lambda$ and $\lambda\leq\opnorm{\what{\Psi}_{\lambda}}\leq M_k+\lambda$ from \cref{appx:chi2_convas}. Moreover, $\what{\Psi}_{\lambda}$ and $\Psi_{\lambda}$ are commuting operators, i.e., $\what{\Psi}_{\lambda}\Psi_{\lambda}f=\Psi_{\lambda}\what{\Psi}_{\lambda}f$ for any $f\in\calK$. 

\paragraph{Convergence rate for \cref{chi2_rate_b}}
Using the result that if two operators $A$ and $B$ are invertible, then $A^{-1}-B^{-1}=B^{-1}(B-A)A^{-1}$, we obtain
\begin{align}
	\abs{\cref{chi2_rate_b}}
	&=\abs*{\ip*{\mu}{\left(\what{\Psi}_{\lambda}^{-1}-\Psi_{\lambda}^{-1}\right)\what{\mu}}_{\calK}}\nn
	&=\abs*{\ip*{\mu}{\Psi_{\lambda}^{-1}\left(\Psi_{\lambda}-\what{\Psi}_{\lambda}\right)\what{\Psi}_{\lambda}^{-1}\what{\mu}}_{\calK}}\nn
	&=\abs*{\ip*{\Psi_{\lambda}^{-1}\mu}{\left(\Psi_{\lambda}-\what{\Psi}_{\lambda}\right)\what{\Psi}_{\lambda}^{-1}\what{\mu}}_{\calK}}. \label{chi2_rate_b_ineq}
\end{align}

Denote $g=\Psi_{\lambda}^{-1}\mu$ and $h=\what{\Psi}_{\lambda}^{-1}\what{\mu}$. Then we have
\begin{align*}
	&\ip*{g}{\what{\Psi}_{\lambda}h}_{\calK}
	=\ofrac{n}\sum_{i=1}^ng(Y_i)h(Y_j),\ \text{and} \nn
	&\ip*{g}{\Psi_{\lambda}h}_{\calK}
	=\E_{Y,Y'}[g(Y)h(Y')].
\end{align*}

Therefore, we can apply McDiarmid's inequality to \cref{chi2_rate_b_ineq}. Next we find $c_i$ for McDiarmid's inequality. 

For any $y$ in $\Omega$, we have
\begin{align*}
	g(y)=\ip*{\Phi(y)}{\Psi_{\lambda}^{-1}\mu}_{\calK}
	&\leq\norm*{\Phi(y)}_{\calK}\norm*{\Psi_{\lambda}^{-1}\mu}_{\calK}\nn
	&=\sqrt{\ip*{\Phi(y)}{\Phi(y)}_{\calK}}\norm*{\Psi_{\lambda}^{-1}\mu}_{\calK}\nn
	&\leq\sqrt{M_k}\alpha.
\end{align*}

Similarly, we have $h(y)\leq\sqrt{M_k}\beta$ for any $y$ in $\Omega$. 
Therefore, we obtain
\begin{align*}
	g(Y_i)h(Y_j)\leq M_k\alpha\beta,
\end{align*}
for any $Y_i$ and $Y_j$. Subsequently, we have $c_i=\dfrac{M_k\alpha\beta}{n}$. Applying McDiarmid's inequality to \cref{chi2_rate_b_ineq}, we arrive at
\begin{align*}
	\P(\cref{chi2_rate_b}\geq t)
	\leq\exp\left(-\dfrac{2nt^2}{M_k^2\alpha^2\beta^2}\right).
\end{align*}

\paragraph{Convergence rate for \cref{chi2_rate_a}}
From the Cauchy–Schwarz inequality, we have
\begin{align}
	\abs*{\ip*{\what{\mu}}{\what{\Psi}_{\lambda}^{-1}\what{\mu}}_{\calK}-\ip*{\mu}{\what{\Psi}_{\lambda}^{-1}\what{\mu}}_{\calK}}
	&=\abs*{\ip*{\mu_{\bbP}-\what{\mu}_{\bbP}}{\what{\Psi}_{\lambda}^{-1}\what{\mu}}_{\calK}+\ip*{\mu_{\bbQ}-\what{\mu}_{\bbQ}}{\what{\Psi}_{\lambda}^{-1}\what{\mu}}_{\calK}}\nn
	&\leq\abs*{\ip*{\mu_{\bbP}-\what{\mu}_{\bbP}}{\what{\Psi}_{\lambda}^{-1}\what{\mu}}_{\calK}}+\abs*{\ip*{\mu_{\bbQ}-\what{\mu}_{\bbQ}}{\what{\Psi}_{\lambda}^{-1}\what{\mu}}_{\calK}}. \label{chi2_rate_a_ineq}
\end{align}

For the R.H.S. of \cref{chi2_rate_a_ineq}, we have
\begin{align*}
	\ip*{\mu_{\bbP}-\what{\mu}_{\bbP}}{\what{\Psi}_{\lambda}^{-1}\what{\mu}}_{\calK}
	&=\ofrac{m}\sum_{i=1}^m\ip*{h}{\Phi(X_i)}_{\calK}-\ip*{h}{\E_{Y}[\Phi(Y)]}_{\calK}\nn
	&=\ofrac{m}\sum_{i=1}^mh(X_i)-\E_{Y}[h(Y)],
\end{align*}
where $h=\what{\Psi}_{\lambda}^{-1}\what{\mu}$. Since $h(y)\leq\sqrt{M_k}\beta$ for any $y\in\Omega$, $c_i=\dfrac{M_k\beta}{m}$ for applying McDiarmid's inequality. Subsequently, we obtain
\begin{align*}
	\P(\abs*{\ip*{\mu_{\bbP}-\what{\mu}_{\bbP}}{\what{\Psi}_{\lambda}^{-1}\what{\mu}}_{\calK}}\geq t)
	\leq\exp\left(-\dfrac{2mt^2}{M_k\beta^2}\right).
\end{align*}

Similarly, we have
\begin{align*}
	\P(\abs*{\ip*{\mu_{\bbQ}-\what{\mu}_{\bbQ}}{\what{\Psi}_{\lambda}^{-1}\what{\mu}}_{\calK}}\geq t)
	\leq\exp\left(-\dfrac{2nt^2}{M_k\beta}\right)
\end{align*}

Combining the above results, we arrive at
\begin{align*}
	\P(\cref{chi2_rate_a}\geq 2t)
	\leq\exp\left(-\dfrac{2mt^2}{M_k\beta^2}\right)+\exp\left(-\dfrac{2nt^2}{M_k\beta^2}\right).
\end{align*}

\paragraph{Convergence rate for \cref{chi2_rate_c}}
Note $\cref{chi2_rate_c}$ has a similar bound as $\cref{chi2_rate_b}$:
\begin{align*}
	\abs*{\ip*{\mu}{\Psi_{\lambda}^{-1}\mu}_{\calK}-\ip*{\mu}{\Psi_{\lambda}^{-1}\what{\mu}}_{\calK}}
	&=\abs*{\ip*{\Psi_{\lambda}^{-1}\mu}{\mu}_{\calK}-\ip*{\Psi_{\lambda}^{-1}\mu}{\what{\mu}}_{\calK}}\nn
	&\leq\abs*{\ip*{g}{\mu_{\bbP}-\what{\mu}_{\bbP}}_{\calK}}+\abs*{\ip*{g}{\mu_{\bbQ}-\what{\mu}_{\bbQ}}_{\calK}},
\end{align*}
where $g=\Psi_{\lambda}^{-1}\mu$. Since $g(y)\leq\sqrt{M_k}\alpha$ for any $y\in\Omega$, $c_i=\dfrac{M_k\alpha}{m}$.
Therefore, $\cref{chi2_rate_c}$ attains a similar convergence rate as $\cref{chi2_rate_b}$:
\begin{align*}
	\P(\cref{chi2_rate_c}\geq 2t)
	\leq\exp\left(-\dfrac{2mt^2}{M_k\alpha^2}\right)+\exp\left(-\dfrac{2nt^2}{M_k\alpha^2}\right).
\end{align*}

\bibliographystyle{unsrtnat}
\bibliography{IEEEabrv,StringDefinitions,BibBooks,refs}

\end{document}

%% file: preamble.tex
\usepackage[T1]{fontenc}
\usepackage[utf8]{inputenc}
\usepackage{amsmath,amssymb,amsfonts,mathrsfs,bm}
\usepackage{mathtools}
\usepackage{scalerel}
\usepackage{nicefrac}
\usepackage[shortlabels]{enumitem}
\usepackage{graphicx}
\usepackage{epstopdf}
\DeclareGraphicsExtensions{.eps,.png,.jpg,.pdf}

\usepackage{url}
\usepackage{colortbl}
\usepackage{booktabs}
\usepackage{multirow}
\usepackage[normalem]{ulem}
\usepackage{xparse}
\usepackage{calc}
\usepackage{etoolbox}

\makeatletter
\@ifpackageloaded{natbib}{
	\relax
}{
	\usepackage{cite}
}
\makeatother

\usepackage{array}
\newcolumntype{L}[1]{>{\raggedright\let\newline\\\arraybackslash\hspace{0pt}}m{#1}}
\newcolumntype{C}[1]{>{\centering\let\newline\\\arraybackslash\hspace{0pt}}m{#1}}
\newcolumntype{R}[1]{>{\raggedleft\let\newline\\\arraybackslash\hspace{0pt}}m{#1}}

\usepackage{glossaries}
\makeatletter
\sfcode`\.1006

\let\oldgls\gls
\let\oldglspl\glspl

\newcommand\fussy@ifnextchar[3]{%
	\let\reserved@d=#1%
	\def\reserved@a{#2}%
	\def\reserved@b{#3}%
	\futurelet\@let@token\fussy@ifnch}
\def\fussy@ifnch{%
	\ifx\@let@token\reserved@d
		\let\reserved@c\reserved@a
	\else
		\let\reserved@c\reserved@b
	\fi
	\reserved@c}

\renewcommand{\gls}[1]{%
\oldgls{#1}\fussy@ifnextchar.{\@checkperiod}{\@}}
\renewcommand{\glspl}[1]{%
\oldglspl{#1}\fussy@ifnextchar.{\@checkperiod}{\@}}

\newcommand{\@checkperiod}[1]{%
	\ifnum\sfcode`\.=\spacefactor\else#1\fi
}
\makeatother

\newacronym{wrt}{w.r.t.}{with respect to}
\newacronym{RHS}{R.H.S.}{right-hand side}
\newacronym{LHS}{L.H.S.}{left-hand side}
\newacronym{iid}{i.i.d.}{independent and identically distributed}

\usepackage{float}

\ifx\notloadhyperref\undefined
	\ifx\loadbibentry\undefined
	\else
		\usepackage{bibentry}
		\makeatletter\let\saved@bibitem\@bibitem\makeatother
		\makeatletter\let\@bibitem\saved@bibitem\makeatother
	\fi
\else
	\ifx\loadbibentry\undefined
		\relax
	\else
		\usepackage{bibentry}
	\fi
\fi

\usepackage[capitalize]{cleveref}
\crefname{equation}{}{}
\Crefname{equation}{}{}
\crefname{claim}{claim}{claims}
\crefname{step}{step}{steps}
\crefname{line}{line}{lines}
\crefname{condition}{condition}{conditions}
\crefname{dmath}{}{}
\crefname{dseries}{}{}
\crefname{dgroup}{}{}

\crefname{Problem}{Problem}{Problems}
\crefformat{Problem}{Problem~(#2#1#3)}
\crefrangeformat{Problem}{Problems~(#3#1#4) to~(#5#2#6)}

\crefname{Theorem}{Theorem}{Theorems}
\crefname{Corollary}{Corollary}{Corollaries}
\crefname{Proposition}{Proposition}{Propositions}
\crefname{Lemma}{Lemma}{Lemmas}
\crefname{Definition}{Definition}{Definitions}
\crefname{Example}{Example}{Examples}
\crefname{Assumption}{Assumption}{Assumptions}
\crefname{Remark}{Remark}{Remarks}
\crefname{Rem}{Remark}{Remarks}
\crefname{remarks}{Remarks}{Remarks}
\crefname{Appendix}{Appendix}{Appendices}
\crefname{Supplement}{Supplement}{Supplements}
\crefname{Exercise}{Exercise}{Exercises}
\crefname{Theorem_A}{Theorem}{Theorems}
\crefname{Corollary_A}{Corollary}{Corollaries}
\crefname{Proposition_A}{Proposition}{Propositions}
\crefname{Lemma_A}{Lemma}{Lemmas}
\crefname{Definition_A}{Definition}{Definitions}

\usepackage{algorithm,algorithmic}

\ifx\loadbreqn\undefined
	\relax
\else
	\usepackage{breqn}
\fi

\makeatletter
\def\cleartheorem#1{%
    \expandafter\let\csname#1\endcsname\relax
    \expandafter\let\csname c@#1\endcsname\relax
}
\def\clearthms#1{ \@for\tname:=#1\do{\cleartheorem\tname} }
\makeatother

\ifx\renewtheorem\undefined
	\ifx\useTheoremCounter\undefined
		\newtheorem{Theorem}{Theorem}
		\newtheorem{Corollary}{Corollary}
		\newtheorem{Proposition}{Proposition}
		\newtheorem{Lemma}{Lemma}
	\else
		\newtheorem{Theorem}{Theorem}
		
		\newtheorem{Proposition}[Theorem]{Proposition}
	\fi

	\newtheorem{Definition}{Definition}

	\newtheorem{Assumption}{Assumption}

\fi

\newcommand{\qednew}{\nobreak \ifvmode \relax \else
		\ifdim\lastskip<1.5em \hskip-\lastskip
			\hskip1.5em plus0em minus0.5em \fi \nobreak
		\vrule height0.75em width0.5em depth0.25em\fi}



\newcommand{\nn}{\nonumber\\ }

\NewDocumentCommand{\movedownsub}{e{^_}}{%
	\IfNoValueTF{#1}{%
		\IfNoValueF{#2}{^{}}
	}{%
		^{#1}
	}%
	\IfNoValueF{#2}{_{#2}}
}

\let\latexchi\chi
\RenewDocumentCommand{\chi}{}{\latexchi\movedownsub}

\newcommand{\calB}{\mathcal{B}}

\newcommand{\calE}{\mathcal{E}}
\newcommand{\calF}{\mathcal{F}}
\newcommand{\calG}{\mathcal{G}}
\newcommand{\calH}{\mathcal{H}}

\newcommand{\calK}{\mathcal{K}}

\newcommand{\calN}{\mathcal{N}}
\newcommand{\calO}{\mathcal{O}}

\newcommand{\bA}{\mathbf{A}}

\newcommand{\bI}{\mathbf{I}}

\newcommand{\bw}{\mathbf{w}}

\newcommand{\bbC}{\mathbb{C}}

\newcommand{\bbN}{\mathbb{N}}

\newcommand{\bbP}{\mathbb{P}}
\newcommand{\bbQ}{\mathbb{Q}}
\newcommand{\bbR}{\mathbb{R}}

\newcommand{\frakJ}{\mathfrak{J}}

\newcommand{\scT}{\mathscr{T}}

\DeclareSymbolFont{bsfletters}{OT1}{cmss}{bx}{n}
\DeclareSymbolFont{ssfletters}{OT1}{cmss}{m}{n}
\DeclareMathSymbol{\bsfGamma}{0}{bsfletters}{'000}
\DeclareMathSymbol{\ssfGamma}{0}{ssfletters}{'000}
\DeclareMathSymbol{\bsfDelta}{0}{bsfletters}{'001}
\DeclareMathSymbol{\ssfDelta}{0}{ssfletters}{'001}
\DeclareMathSymbol{\bsfTheta}{0}{bsfletters}{'002}
\DeclareMathSymbol{\ssfTheta}{0}{ssfletters}{'002}
\DeclareMathSymbol{\bsfLambda}{0}{bsfletters}{'003}
\DeclareMathSymbol{\ssfLambda}{0}{ssfletters}{'003}
\DeclareMathSymbol{\bsfXi}{0}{bsfletters}{'004}
\DeclareMathSymbol{\ssfXi}{0}{ssfletters}{'004}
\DeclareMathSymbol{\bsfPi}{0}{bsfletters}{'005}
\DeclareMathSymbol{\ssfPi}{0}{ssfletters}{'005}
\DeclareMathSymbol{\bsfSigma}{0}{bsfletters}{'006}
\DeclareMathSymbol{\ssfSigma}{0}{ssfletters}{'006}
\DeclareMathSymbol{\bsfUpsilon}{0}{bsfletters}{'007}
\DeclareMathSymbol{\ssfUpsilon}{0}{ssfletters}{'007}
\DeclareMathSymbol{\bsfPhi}{0}{bsfletters}{'010}
\DeclareMathSymbol{\ssfPhi}{0}{ssfletters}{'010}
\DeclareMathSymbol{\bsfPsi}{0}{bsfletters}{'011}
\DeclareMathSymbol{\ssfPsi}{0}{ssfletters}{'011}
\DeclareMathSymbol{\bsfOmega}{0}{bsfletters}{'012}
\DeclareMathSymbol{\ssfOmega}{0}{ssfletters}{'012}

\newcommand{\btheta}{\bm{\theta}}

\newcommand{\bSigma	}{\bm{\Sigma}}

\newcommand{\bPhi}{\bm{\Phi}}

\makeatletter
\newcommand*\rel@kern[1]{\kern#1\dimexpr\macc@kerna}
\newcommand*\widebar[1]{%
  \begingroup
  \def\mathaccent##1##2{%
    \rel@kern{0.8}%
    \overline{\rel@kern{-0.8}\macc@nucleus\rel@kern{0.2}}%
    \rel@kern{-0.2}%
  }%
  \macc@depth\@ne
  \let\math@bgroup\@empty \let\math@egroup\macc@set@skewchar
  \mathsurround\z@ \frozen@everymath{\mathgroup\macc@group\relax}%
  \macc@set@skewchar\relax
  \let\mathaccentV\macc@nested@a
  \macc@nested@a\relax111{#1}%
  \endgroup
}
\makeatother

\DeclareMathOperator*{\argsup}{arg\,sup}

\DeclareMathOperator{\Tr}{Tr}
\DeclareMathOperator{\spn}{span}

\DeclareMathOperator{\var}{var}

\DeclareMathOperator{\cov}{cov}

\newcommand{\ifbcdot}[1]{\ifblank{#1}{\cdot}{#1}}

\DeclarePairedDelimiterX\abs[1]{\lvert}{\rvert}{\ifbcdot{#1}}
\DeclarePairedDelimiterX\parens[1]{(}{)}{\ifbcdot{#1}}
\DeclarePairedDelimiterX\brk[1]{[}{]}{\ifbcdot{#1}}
\DeclarePairedDelimiterX\braces[1]{\{}{\}}{\ifbcdot{#1}}
\DeclarePairedDelimiterX\angles[1]{\langle}{\rangle}{\ifblank{#1}{\cdot,\cdot}{#1}}
\DeclarePairedDelimiterX\ip[2]{\langle}{\rangle}{\ifbcdot{#1},\ifbcdot{#2}}
\DeclarePairedDelimiterX\norm[1]{\lVert}{\rVert}{\ifbcdot{#1}}
\DeclarePairedDelimiterX\ceil[1]{\lceil}{\rceil}{\ifbcdot{#1}}
\DeclarePairedDelimiterX\floor[1]{\lfloor}{\rfloor}{\ifbcdot{#1}}

\DeclarePairedDelimiterXPP\trace[1]{\operatorname{Tr}}{(}{)}{}{\ifbcdot{#1}} 
\DeclarePairedDelimiterXPP\col[1]{\operatorname{col}}{\{}{\}}{}{\ifbcdot{#1}} 
\DeclarePairedDelimiterXPP\row[1]{\operatorname{row}}{\{}{\}}{}{\ifbcdot{#1}} 
\DeclarePairedDelimiterXPP\erf[1]{\operatorname{erf}}{(}{)}{}{\ifbcdot{#1}}
\DeclarePairedDelimiterXPP\erfc[1]{\operatorname{erfc}}{(}{)}{}{\ifbcdot{#1}}
\DeclarePairedDelimiterXPP\KLD[2]{D_f}{(}{)}{}{\ifbcdot{#1}\, \delimsize\|\, \ifbcdot{#2}} 
\DeclarePairedDelimiterXPP\op[2]{\operatorname{#1}}{(}{)}{}{#2} 

\newcommand{\convas}{\stackrel{\mathrm{a.s.}}{\longrightarrow}}
\newcommand{\convd}{\stackrel{\mathrm{d}}{\longrightarrow}}

\newcommand{\T}{^{\intercal}}
\newcommand{\ud}{\,\mathrm{d}} 

\DeclarePairedDelimiterXPP\indicate[1]{{\bf 1}}{\{}{\}}{}{\ifbcdot{#1}}

\newcommand{\ofrac}[1]{{\frac{1}{#1}}}
\newcommand{\odfrac}[1]{{\dfrac{1}{#1}}}
\newcommand{\ddfrac}[2]{{\dfrac{\mathrm{d} {#1}}{\mathrm{d} {#2}}}}
\newcommand{\ppfrac}[2]{\dfrac{\partial {#1}}{\partial {#2}}}

\providecommand\given{}

\DeclarePairedDelimiterX\Set[2]\{\}{%
\renewcommand\given{\SetSymbol[\delimsize]{#1}}
#2
}
\DeclarePairedDelimiterX\Setc[1]\{\}{%
\renewcommand\given{\SetSymbol{:}}
#1
}

\NewDocumentCommand\set{s m t| m}{%
\IfBooleanTF#1%
{\left\{\, #2\mathrel{} \IfBooleanTF{#3}{\middle|}{:}\mathrel{}  #4\, \right\}}%
{\{\, #2 \IfBooleanTF{#3}{\mid}{\mathrel{} : \mathrel{}} #4\, \}}%
}

\NewDocumentCommand{\evalat}{ s O{\big} m e{_^} }{%
\IfBooleanTF{#1}%
{\left. #3 \right|}{#3#2|}%
\IfValueT{#4}{_{#4}}%
\IfValueT{#5}{^{#5}}%
}

\providecommand\given{}
\DeclarePairedDelimiterXPP\cprob[1]{}(){}{
\renewcommand\given{\nonscript\,\delimsize\vert\allowbreak\nonscript\,\mathopen{}}%
#1%
}
\DeclarePairedDelimiterXPP\cexp[1]{}[]{}{
\renewcommand\given{\nonscript\,\delimsize\vert\allowbreak\nonscript\,\mathopen{}}%
#1%
}

\DeclareDocumentCommand \P { s e{_^} d() g } {%
	\mathbb{P}%
	\IfBooleanTF{#1}%
		{
			\IfValueT{#2}{_{#2}}%
			\IfValueT{#3}{^{#3}}%
			\IfValueTF{#5}{\cprob{#4 \given #5}}{\IfValueT{#4}{\cprob{#4}}}%
		}%
		{
			\IfValueT{#2}{_{#2}}%
			\IfValueT{#3}{^{#3}}%
			\IfValueTF{#5}{\cprob*{#4 \given #5}}{\IfValueT{#4}{\cprob*{#4}}}%
		}%
}

\DeclareDocumentCommand \E { s e{_^} o g } {%
	\mathbb{E}%
	\IfBooleanTF{#1}%
		{
			\IfValueT{#2}{_{#2}}%
			\IfValueT{#3}{^{#3}}%
			\IfValueTF{#5}{\cexp{#4 \given #5}}{\IfValueT{#4}{\cexp{#4}}}%
		}%
		{
			\IfValueT{#2}{_{#2}}%
			\IfValueT{#3}{^{#3}}%
			\IfValueTF{#5}{\cexp*{#4 \given #5}}{\IfValueT{#4}{\cexp*{#4}}}%
		}%
}

\ExplSyntaxOn
\NewDocumentCommand \dist {m o o} {%
\mathrm{#1}\left(%
	\IfValueT{#3}{%
		\tl_if_blank:nTF{ #3 }{\cdot\, \middle|\, }{#3\, \middle|\, }%
	}
	\IfValueT{#2}{#2}%
\right)%
}
\ExplSyntaxOff

\newcommand{\N}[2]{\dist{\calN}[#1,\, #2]}

\NewDocumentCommand {\cbrace} {t+ D[]{black} D(){\widthof{#5}} m m } {%
	\begingroup%
		\color{#2}
		\IfBooleanTF{#1}{%
			\overbrace{#4}^%
		}{
			\underbrace{#4}_%
		}%
		{\parbox[c]{#3}{\centering\footnotesize{#5}}}%
	\endgroup%
}

\let\oldforall\forall
\renewcommand{\forall}{\oldforall \, }

\let\oldexist\exists
\renewcommand{\exists}{\oldexist \, }

\DeclareDocumentCommand{\includeCroppedPdf}{ o O{./Figures/} m }{
	\IfFileExists{#2#3-crop.pdf}{}{%
		\immediate\write18{pdfcrop #2#3.pdf #2#3-crop.pdf}}%
	\includegraphics[#1]{#2#3-crop.pdf}
}

\makeatletter
\newcommand*{\addFileDependency}[1]{
  \typeout{(#1)}
  \@addtofilelist{#1}
  \IfFileExists{#1}{}{\typeout{No file #1.}}
}
\makeatother

\definecolor{gray90}{gray}{0.9}

\ifx\nohighlights\undefined

	\newcommand{\msout}[1]{\text{\color{green} \sout{\ensuremath{#1}}}}
	\newcommand{\del}[1]{{\color{green}\ifmmode \msout{#1}\else\sout{#1}\fi}}
\else

	\newcommand{\msout}[1]{#1}
	\newcommand{\del}[1]{#1}
\fi

\newcommand{\hhide}[1]{}

\ifx\diagnoselabel\undefined
	\relax
\else
	\makeatletter
	\def\@testdef #1#2#3{%
		\def\reserved@a{#3}\expandafter \ifx \csname #1@#2\endcsname
			\reserved@a  \else
			\typeout{^^Jlabel #2 changed:^^J%
				\meaning\reserved@a^^J%
				\expandafter\meaning\csname #1@#2\endcsname^^J}%
			\@tempswatrue \fi}
	\makeatother
\fi
